\documentclass[10pt]{elsarticle}

\usepackage[T1]{fontenc}
\usepackage[utf8]{inputenc}
\usepackage{amsmath,amsthm,amsfonts,amssymb,amscd}
\usepackage[numbers]{natbib}
\usepackage[colorlinks,citecolor=blue,urlcolor=blue]{hyperref}
\usepackage{xcolor,graphicx,subcaption,booktabs,multicol,siunitx,threeparttable,multirow,latexsym,array,hhline,float,appendix,caption,epstopdf,tkz-tab}
\usepackage{circuitikz,tikz}

\everymath{\displaystyle}

\newcommand{\cA}{{\cal A}}  
\newcommand{\cF}{{\cal F}}  
\newcommand{\cB}{{\cal B}} \newcommand{\cH}{{\cal H}}

\newcommand{\cP}{{\cal P}}

\newcommand{\R}{\mathbb{R}} \newcommand{\N}{\mathbb{N}}

\DeclareMathOperator{\essinf}{ess\,inf}
\DeclareMathOperator{\esssup}{ess\,sup}
\DeclareMathOperator{\supp}{supp}

\newtheorem{theo}{Theorem}[section]
\newtheorem{prop}[theo]{Proposition}
\newtheorem{lemm}[theo]{Lemma}
\newtheorem{defi}[theo]{Definition}

\newtheorem{coro}[theo]{Corollary}
\newtheorem{rem}[theo]{Remark}
\newtheorem{Proof}[theo]{Proof}
\newcommand{\fdem}{\hfill$\Box$} % end-of-proof symbol aligned to right

\newcommand{\beq}{\begin{equation}}
\newcommand{\eeq}{\end{equation}}
\newcommand{\beqa}{\begin{equation*}}
\newcommand{\eeqa}{\end{equation*}}
\newcommand{\bea}{\begin{eqnarray}}
\newcommand{\eea}{\end{eqnarray}}
\newcommand{\bean}{\begin{eqnarray*}}
\newcommand{\eean}{\end{eqnarray*}}

\begin{document}

%%%%%%%%%%%%%%%%%%% FRONTMATTER %%%%%%%%%%%%%%%%%%%%%%%%

\begin{frontmatter}

\title{Explicit Recursive Construction of Super-Replication Prices under Proportional Transaction Costs}

\author[A1]{Emmanuel LEPINETTE}
\author[A1]{Amal OMRANI}

\address[A1]{CEREMADE, UMR CNRS 7534, Paris-Dauphine University, PSL,\\
Place du Maréchal De Lattre De Tassigny, 75775 Paris Cedex 16, France.\\
Emails: emmanuel.lepinette@ceremade.dauphine.fr, omrani@ceremade.dauphine.fr}

\begin{abstract}
We propose a constructive framework for the super-hedging problem of a European contingent claim under proportional transaction costs in discrete time. 
Our main contribution is an explicit recursive scheme that computes both the super-hedging price and the corresponding optimal strategy without relying on martingale arguments. 
The method is based on convex duality and a distorted Legendre--Fenchel transform, ensuring both tractability and convexity of the value functions. 
A numerical implementation on real market data illustrates the practical
relevance of the proposed approach.
\end{abstract}


\end{frontmatter}
% ===== Section 1: Introduction =====
\section{Introduction}
Let $(\Omega,(\mathcal{F}_t)_{t=0,\dots,T},\mathbb{P})$ be a complete filtered probability space in discrete time, with maturity date $T$ for a European contingent claim. We consider a financial market consisting of a risk-free bond $S^0 = 1$ (w.l.o.g.) and a risky asset described by a positive adapted stochastic process $S=(S_t)_{t=0,\dots,T}$. Trading strategies are given by adapted processes $\phi=(\phi_t)_{t=0,\dots,T}$ representing the number of units held in each asset. For any stochastic process $X=(X_t)_{t=0,\dots,T}$, we denote $\Delta X_t := X_t - X_{t-1}$, $t\ge 1$, and random variables are defined up to negligible sets.

We consider a deterministic sequence of proportional transaction costs $(\kappa_t)_{t=0}^{T-1}$, with $\kappa_t \in [0,1)$, representing the transaction costs rate for trades between times $t$ and $t+1$. A self-financing portfolio $(V_t)_{t=0,\dots,T}$ satisfies by definition
\begin{equation} \label{Dyn}
\Delta V_t = \phi_{t-1}\Delta S_t - \kappa_{t-1}|\Delta \phi_{t-1}| S_{t-1}, \quad t=1,\dots,T,
\end{equation}
where the term $\kappa_{t-1}|\Delta \phi_{t-1}| S_{t-1}$ represents the transaction costs for adjusting the strategy from $\phi_{t-2}$ to $\phi_{t-1}$. The portfolio value $V_t$ is interpreted as the super-hedging price of a payoff $\xi_T$ at time $t$ if $V_T \ge \xi_T$ almost surely. 

The model we consider is defined by deterministic coefficients $\alpha_{t-1}$ and $\beta_{t-1}$ such that $0 \le \alpha_{t-1} < \beta_{t-1}$ and the conditional support of $\frac{S_t}{S_{t-1}}$ is given by
\bea \label{MS}
{\rm supp}_{\mathcal{F}_{t-1}} \frac{S_t}{S_{t-1}} = [\alpha_{t-1}, \beta_{t-1}],\quad t=1,\cdots,T.
\eea

The computation of super-hedging prices under proportional transaction costs is well known to be challenging.
Continuous-time asymptotic approaches, beginning with Leland (1985) \cite{Leland85} and followed by Kabanov and Safarian (1997) \cite{KS}, provide approximate hedging strategies that account for transaction costs.
Kabanov and Lepinette \cite{KL}, \cite{Lep0} further analyzed the mean-square error of Leland-type strategies for convex payoffs.
Permenshchikov \cite{Perg} established limit theorems for Leland’s strategy, and later, together with Thai  \cite{PT1,PT2}, investigated approximate hedging in markets with jumps and stochastic volatility. Also, some asymptotics are given in \cite{Albanese06}, in the setting of Leland's framework for small transaction costs while the case of fixed costs is solved in \cite{Lep1}. A recent study by Biagini et al. \cite{Biagini23} proposes the use of neural networks to approximate super-hedging prices in incomplete markets, within the classical framework based on martingale probability measures.The effect of transaction costs can also be modeled through price impacts, following the approach in \cite{WG}.

In discrete time, only a few numerical schemes have been proposed to compute super-hedging prices and corresponding optimal strategies under transaction costs.
For instance, Föllmer and Schweizer  \cite{FK97}, Bouchard et al.  \cite{BSV03}, Rouge and El Karoui \cite{Rouge}, and Gobet and Miri  \cite{GM03} develop methods within proportional transaction costs settings.
These approaches typically rely on the existence of a risk-neutral measure, on restricted asset dynamics (e.g., binomial or specific diffusion models), or on asymptotic approximations as the number of trading dates becomes large.
Consequently, the resulting prices are often approximate, asymptotic, or model-specific, which limits their applicability for a fixed, moderate number of trading dates or for general discrete-time markets.

In the framework of the Kabanov model with proportional transaction costs, dual formulations of the super-hedging prices have been established under strong no-arbitrage conditions ensuring the existence of dual elements, called consistent price systems; see, for instance, \cite{KS,GLR,DKL,CS,Koehl01}.
However, these results remain mostly theoretical and cannot be implemented in practice.
 
In contrast, our method provides a computational scheme that can estimate exact super-hedging prices for a fixed number of trading dates, without assuming a risk-neutral measure or full no-arbitrage conditions. Only a weak assumption is required, ensuring that the infimum super-hedging price of any non-negative payoff is itself non-negative. The method delivers prices valid  for any proportional transaction costs rate and simultaneously yields the corresponding optimal trading strategy at each date.

The main idea is to propagate backward the pricing method. For a convex payoff $g_T(S_T)$ at time $T$, we show that at time $T-1$ there exists a minimal super-hedging price $P_{T-1}^*(\phi_{T-2})$ depending on the strategy $\phi_{T-2}$ chosen at time $T-2$. Iterating this procedure defines a sequence of convex functions $g_t(\phi_{t-1}, S_t)$, $t = T-1, \dots, 0$, characterizing both the super-hedging prices and the optimal strategies at each date.

Technically, the approach relies on distorted Legendre-Fenchel conjugates, extending the classical conjugate method used in frictionless markets \cite{CL}, \cite{ElMansour22}. This allows the explicit computation of $P_t^*(\phi_{t-1})$ as a function of the previous strategy, a novelty in the presence of transaction costs. Furthermore, the method can be extended to non-convex payoffs as well as to Asian and American options, although the computations become more involved. It is important to note that our method is highly innovative and relies on optimization techniques. For instance, \cite{VZ} also adopt an original approach, rather than the classical martingale-based framework.

The paper is organized as follows. Section \ref{MR} presents the main results. Section \ref{NI} shows an example of implementation on real data from the US index S\&P 500. In Section \ref{section propagation}, the general one step method is developed. In Section \ref{A}, some proofs are postponed.

% ===== Section 2 Main results =====
\section{Main results}\label{MR}
This section provides explicit recursive representations of the minimal superhedging prices in the presence of proportional transaction costs. Two distinct regimes are considered, depending on whether transaction costs are "large" or "small". The  computation of superhedging prices  over the time interval $[0, T]$ is based on a backward induction scheme: starting from the known terminal payoff at time $T$, the problem is solved step by step backward in time. At each intermediate date, this construction yields both the minimal superhedging price and the corresponding optimal trading strategy, which depend on the transaction costs rate and the earlier strategy.

Subsection~\ref{GeneralStepProc1} presents the solution of the general one-step problem, providing explicit formulas and conditions under which the infimum is attained. Subsection~\ref{MSP} extends this construction to the multi-step problem, leading to a complete characterization of the superhedging prices and associated optimal strategies over the entire time horizon. 

% ===== Subsection: One step superhedging price =====
\subsection{The one step infimum superhedging price for a European claim}\label{GeneralStepProc1}
\bigskip

Suppose that, at time $t$, the payoff function depends on $\phi_{t-1} \in L^0(\R,\cF_{t-1})$, and, for some $N_t \in \mathbb{N}$, it admits the following form:
\bea\label{GeneralPayoffForm1} 
g_t(\phi_{t-1}, x) = \max_{i=1,\dots,N_t} g_t^i(\phi_{t-1}, x),
\eea
where, for each $i=1,\dots,N_t$, the mapping $x \mapsto g_t^i(\phi_{t-1}, x)$ is convex and satisfies
\bea \label{formGi1} 
g_t^i(\phi_{t-1}, x) = \hat{g}_t^i(x) - \hat{\mu}_t^i \phi_{t-1} x,
\eea
where $(\hat{\mu}_t^i)_{i=1}^{N_t}$ is a deterministic and strictly increasing sequence such that 
\bea \label{formGi-11}  
1 + \hat{\mu}_t^i > 0, \quad i=1,\dots,N_t,
\eea  
while the functions $\hat{g}_t^i$ are independent of $\phi_{t-1}$.\footnote{For $t=T$, we have $N_T=1$, $\hat\mu_T^1=0$, and $\hat g_T=g_T$. If $\hat \mu_t^i=\hat \mu_t^{i+1}$, both $g_t^i$ and $g_t^{i+1}$ can be replaced by $(\phi_{t-1},x)\mapsto\max(\hat g_t^i(x), \hat g_t^{i+1}(x))-\hat \mu_t^i \phi_{t-1}x$.}
\smallskip

We introduce the scaled parameters (see \ref{MS}) and auxiliary functions:
\bea
&& \alpha_{t-1}^i = \alpha_{t-1}(1+\hat\mu_t^i), \quad i=1,\dots,N_t,\\
&& \beta_{t-1}^i = \beta_{t-1}(1+\hat\mu_t^i), \quad i=1,\dots,N_t,\\
&& \tilde{g}_{t-1}^i(x) =
\begin{cases}
\hat{g}_t^i(\alpha_{t-1} x), & i=1,\dots,N_t,\\[1mm]
\hat{g}_t^{i-N_t}(\beta_{t-1} x), & i=N_t+1,\dots,2N_t.
\end{cases}
 \label{tilde-g}
\eea

The following definition provides the minimal no-arbitrage condition under which our results are formulated.

\begin{defi}\label{AIPrel} 
An \emph{immediate profit} at time $t-1$ \emph{relatively to} the payoff function $g_t$  is an infimum price $p_{t-1}(g_t)$  satisfying $\mathbb{P}(p_{t-1}(g_t)=-\infty)>0$ and with $\theta_{t-2}=0$. 
We say that the relative AIP condition holds for $g_t$ at time $t-1$ if there exists no immediate profit at time $t-1$ relative to $g_t$.
\end{defi}

The following result is a consequence of Theorem \ref{optimization}.

\begin{prop}\label{AIP-tminus1}  
 Let $g_t$ be a convex payoff function of the form \eqref{GeneralPayoffForm1}.
The AIP condition holds at time $t-1$ relatively to $g_t$ if and only if
\begin{equation}\label{AIPC}
  \alpha_{t-1}^1  \leq 1+\kappa_{t-1} \quad \text{and} \quad \beta^{N_t}_{t-1}  \geq 1- \kappa_{t-1}.  
\end{equation}

\end{prop}

\begin{rem}
This result is already known for $\kappa_{t-1}=0$, see \cite{CL}. As expected, the AIP condition holds as soon as the transaction costs rate is large enough. 
\end{rem}

The proof of the following theorem is given in Proof \ref{Pr-Prop-Conv}. It shows that the infimum super-hedging price at time $t-1$ is still a payoff function of $S_{t-1}$ in the same structural form as $g_t$. 

\begin{theo}\label{theo:payoff}
 Let $g_t$ be a convex payoff function of the form \eqref{GeneralPayoffForm1}. Suppose that the AIP condition of Definition~\ref{AIPrel} holds. 
Then the infimum superhedging price at time $t-1$ satisfies
\[
p_{t-1}(g_t) = g_{t-1}(\phi_{t-2}, S_{t-1}),
\]
where $g_{t-1}$ is a payoff function of the same structural form, satisfying~\eqref{formGi1}–\eqref{formGi-11} with $t$ replaced by $t-1$.
\end{theo}

\bigskip

The following propositions are established under the AIP condition introduced in Definition~\ref{AIPrel}.
The equivalent AIP condition  of Proposition~\ref{AIP-tminus1} naturally separates into two distinct subcases that must be analyzed individually.
These results are derived from the proof of Theorem~\ref{theo:payoff}, see Proof \ref{Pr-Prop-Conv}, and provide an explicit representation of the payoff function at time $t-1$ in terms of the payoff function at time $t$, the transaction costs rate, and the coefficients describing the conditional support of $S_t/S_{t-1}$, see (\ref{MS}).

We introduce the following random function $a_{t-1}$, which plays a key role in characterizing the optimal strategy corresponding to the minimal superhedging price presented below:
\begin{equation}\label{Deltahedging}
a_{t-1}(v) := \max_{j = 2, \dots, 2N_t} \big( b_{t-1}^j - p_{t-1}^j v \big),
\end{equation}
where
\[
\begin{aligned}
p_{t-1}^j &= 
\begin{cases}
\dfrac{1}{(\alpha_{t-1}^j - \alpha_{t-1}^1) S_{t-1}}, & j = 2, \dots, N_t,\\[2mm]
\dfrac{1}{(\beta_{t-1}^{j - N_t} - \alpha_{t-1}^1) S_{t-1}}, & j = N_t + 1, \dots, 2N_t,
\end{cases} \\
b_{t-1}^j &= 
\begin{cases}
\dfrac{\tilde g_{t-1}^j(S_{t-1}) - \tilde g_{t-1}^1(S_{t-1})}{(\alpha_{t-1}^j - \alpha_{t-1}^1) S_{t-1}}, & j = 2, \dots, N_t,\\[2mm]
\dfrac{\tilde g_{t-1}^j(S_{t-1}) - \tilde g_{t-1}^1(S_{t-1})}{(\beta_{t-1}^{j - N_t} - \alpha_{t-1}^1) S_{t-1}}, & j = N_t + 1, \dots, 2N_t.
\end{cases}
\end{aligned}
\]
We suppose without loss of generality that the sequence $(p_{t-1}^j)_{j=2,\cdots,2N_t}$ is non decreasing.

\begin{rem}
Note that the quantities $\bar b_{t-1}^j$ may be interpreted as  Delta-hedging strategies, which are adjusted in the expressions of $a_{t-1}$ by $p_{t-1}^j$ to take into account  the transaction costs.\
\end{rem}

Under the AIP condition \eqref{AIPC}, we distinguish two subcases (Cases 1 and 2):

   \begin{enumerate}
     \item  \label{case:large} Large transaction costs : \(\alpha_{t-1}^1 \ge 1-\kappa_{t-1} \) (i.e. $\kappa_{t-1}\ge 1+\alpha_{t-1}^1$). %     Prices cannot drop enough to make selling attractive, so you never reduce your stock position. You buy only if the marginal benefit of holding more (given the convex payoff and future uncertainty) exceeds the transaction cost.

     \item  \label{case:mild} Small transaction costs: \(\alpha_{t-1}^1 \le 1-\kappa_{t-1} \). %The lowest possible price tomorrow allows for profitable trading, so you may either buy or sell depending on the marginal benefit relative to transaction costs.
\end{enumerate}

\subsubsection{Minimal superhedging price and optimal strategy in Case \ref{case:large}}
\smallskip

Proposition \ref{Prop-Conv11} below provides an explicit expression for the minimal superhedging price in the large transaction costs regime.  
To do so, we  define the following parameters:
\begin{align*}
w_{t-1}^i =
\begin{cases}
     1, & i = 1, \\[1mm]
     1 - \dfrac{\rho_{t-1}}{\alpha_{t-1}^i - \alpha_{t-1}^1}, & i = 2, \dots, N_t,\\[2mm]
     1 - \dfrac{\rho_{t-1}}{\beta_{t-1}^{i-N_t} - \alpha_{t-1}^1}, & i = N_t + 1, \dots, 2N_t,
\end{cases}
\end{align*}
where $\rho_{t-1} = (1 + \kappa_{t-1}) - \alpha_{t-1}^1 \ge 0.$ We deduce the following sets:
\[
 \mathcal{I}_{11} = \{ j : w_{t-1}^j > 0 \}, \qquad 
 \mathcal{I}_{12} = \{ (i,j) : w_{t-1}^i \le 0,\, w_{t-1}^j > 0 \}.
\]
At last, let us define
\bea\label{mu111}\notag
\lambda_{t-1}^{i,j} &=& \frac{|w_{t-1}^j|(1 - w_{t-1}^i)}{|w_{t-1}^j - w_{t-1}^i|} \in [0,1], \quad (i,j) \in \mathcal{I}_{12}, \\
\hat\mu_{t-1}^{i,j} &=& \kappa_{t-1} + \frac{w_{t-1}^i \rho_{t-1}}{1 - w_{t-1}^i} \mathbf{1}_{\{ j = 1 \}} > -1, 
\quad (i,j) \in \mathcal{I}_{12},\\
\hat\mu_{t-1}^i &=& (\rho_{t-1} - \kappa_{t-1}) \mathbf{1}_{\{ i = 1 \}} + \kappa_{t-1} \mathbf{1}_{\{ i \ne 1 \}} > -1,
\quad i \in \mathcal{I}_{11}.
\label{mu121}
\eea

\begin{prop}\label{Prop-Conv11}
Let \(g_t\) be a convex payoff function of the form \eqref{GeneralPayoffForm1}, and suppose that the AIP condition holds with 
\(\alpha_{t-1}^1 \geq 1-\kappa_{t-1} \). 
Then, the minimal super-hedging price at time \(t-1\) is given by
\footnote{We adopt the convention $\max(\emptyset) = -\infty$.}
\[
p_{t-1}(g_t) = g_{t-1}(\phi_{t-2}, S_{t-1}) 
= \max_{(i,j) \in \mathcal{I}_{12}} g_{t-1}^{i,j}(\phi_{t-2}, S_{t-1})
\;\vee\; \max_{i \in \mathcal{I}_{11}} g_{t-1}^i(\phi_{t-2}, S_{t-1}),
\]
where
\[
\begin{aligned}
g_{t-1}^{i,j}(\phi_{t-2}, x) &= \hat g_{t-1}^{i,j}(x) - \hat\mu_{t-1}^{i,j} \phi_{t-2} x, 
&\quad 
\hat g_{t-1}^{i,j}(x) &= \lambda_{t-1}^{i,j} \tilde g_{t-1}^i(x) + (1 - \lambda_{t-1}^{i,j}) \tilde g_{t-1}^j(x),\\[1mm]
g_{t-1}^i(\phi_{t-2}, x) &= \hat g_{t-1}^i(x) - \hat\mu_{t-1}^i \phi_{t-2} x, 
&\quad 
\hat g_{t-1}^i(x) &= w_{t-1}^i \tilde g_{t-1}^1(x) + (1 - w_{t-1}^i) \tilde g_{t-1}^i(x).
\end{aligned}
\]
\end{prop}

\begin{rem}
The function $g_{t-1}$ retains the same structural form as $g_t$ in \eqref{GeneralPayoffForm1}, featuring convex components and an affine dependence on $\phi_{t-2}$. This property ensures that the backward induction may be carried out recursively.
\end{rem}

To express the optimal strategy associated to the minimal price above, we introduce the following quantities at time $t-1$:
\[
\begin{aligned}
d_{t-1}^1 &= \phi_{t-2}\,(1-\alpha_{t-1}^1+\kappa_{t-1})\,S_{t-1}, 
&\quad s_{t-1}^1 &= 1,\\[1mm]
d_{t-1}^i &= (1-\alpha_{t-1}^1+\kappa_{t-1})\,S_{t-1}\,b_{t-1}^i, 
&\quad s_{t-1}^i &= 1 - (1-\alpha_{t-1}^1+\kappa_{t-1})\,S_{t-1}\,p_{t-1}^i,
\\[1mm]
A_{t-1}^i(v) &= d_{t-1}^i + s_{t-1}^i\,v, 
&\quad A_{t-1}^{(1)}(v) &= \max_{i=1,\dots,2N_t} A_{t-1}^i(v).
\end{aligned}
\]
\[
I_{i,j}\text{ solves } A^i_{t-1}(I_{i,j})=A^j_{t-1}(I_{i,j}),\quad i,j=1,\dots,2N.
\]
Note that we only need the points $I_{i,j}$ when existence holds. The proof of the following is given in Appendix \ref{Pr-OptStrat11}, see Proof \ref{Pr-OptStrat11}.
\begin{prop}\label{OptStrat11}     Let \(g_t\) be a convex payoff function of the form \eqref{GeneralPayoffForm1}, and suppose that the AIP condition holds with 
\(\alpha_{t-1}^1 \geq 1-\kappa_{t-1} \) . An optimal strategy associated with the minimal superhedging price $p_{t-1}(g_t)$ given in Proposition \ref{Prop-Conv11} is \[ \phi_{t-1}^{\mathrm{opt}} = a_{t-1}(v^*_{t-1})\vee\phi_{t-2}, \] where  \bean &v^*_{t-1}\in \arg\min_{e\in I_{t-1}} \big|A_{t-1}^{(1)}(e)-p_{t-1}(g_t)+\hat g^1_t(\alpha_{t-1}S_{t-1})+\kappa_{t-1}\phi_{t-2}S_{t-1}\big|,\\&I_{t-1} = \bigl\{ I_{i,j} : i\ne j,\; s_{t-1}^i\le 0,\; s_{t-1}^j\ge 0 \bigr\}. 
 \eean 
    
    \end{prop}
\begin{rem} This result shows that, in Case~\ref{case:large}, the quantity of risky assets at time $t-1$ should be increased, i.e.  $\phi_{t-1}^{\mathrm{opt}} \ge \phi_{t-2}$,  when $\alpha_{t-1}^1 \ge 1 - \kappa_{t-1}$.
In the absence of transaction costs, Case~\ref{case:large} simplifies to $\alpha_{t-1} = 1$, confirming that the optimal strategy is to increase the risky position, as $S_{t} \ge S_{t-1}$. More generally, Case~\ref{case:large} can be rewritten as $S_{t-1}\alpha_{t-1}^1\ge  S_{t-1}^b$ where $S_{t-1}^b=S_{t-1}(1-k_{t-1})$ represents the bid-price while $S_{t-1}\alpha_{t-1}^1$ may be viewed as the "minimal" feasible value of $S_t$ reflecting the effect of the transaction costs.
 \end{rem} %When $\alpha_{t-1}^1 \ge 1 - \kappa_{t-1},$ prices cannot drop enough to make selling attractive, so you never reduce your stock position.  You buy only if the marginal benefit of holding more (given the convex payoff and future uncertainty) exceeds the transaction cost — mathematically, when  $a_{t-1}(v^*) > \varphi_{t-2}.

\subsubsection{Minimal superhedging price and optimal strategy in Case~\ref{case:mild}}
\smallskip

As in the first case, we need to consider some  coefficients.  For \( r = 1, 2 \), 
\begin{align*}
w_{t-1}^{(r,i)} &=
\begin{cases}
1 - \dfrac{\rho_{t-1}^r}{\alpha_{t-1}^i - \alpha_{t-1}^1}, & i = 2, \dots, N_t,\\[1mm]
1 - \dfrac{\rho_{t-1}^r}{\beta_{t-1}^{\,i - N_t} - \alpha_{t-1}^1}, & i = N_t + 1, \dots, 2N_t,
\end{cases}
\end{align*}
where \( \rho_{t-1}^r = (1 - (-1)^r \kappa_{t-1}) - \alpha_{t-1}^1 \), $r=1,2$. We then deduce  the  sets
\begin{align*}
\mathcal{I}_{21} &=
    \bigl\{ (r,i) : w_{t-1}^{(r,i)} > 0 \bigr\},\\ 
\mathcal{I}_{22} &=
    \bigl\{ (r,l,i,j) :
    w_{t-1}^{(r,i)} \le 0,\; w_{t-1}^{(l,j)} > 0 \bigr\}.
\end{align*}

Finally, we deduce  the finite family of parameters $\hat{\mu}_{t-1}^m$, $m=(r,l,i,j)$, at time $t-1$ and some coefficients $\lambda_{t-1}^{m}$ allowing us to define the functions $\hat g_{t-1}^m$ at time $t-1$. Precisely, we have:
\begin{align}
\lambda_{t-1}^{(r,l,i,j)} &= 
\frac{|w_{t-1}^{(l,j)}| \bigl(1 - w_{t-1}^{(r,i)}\bigr)}
{|w_{t-1}^{(l,j)} - w_{t-1}^{(r,i)}|} \in [0,1],\quad (r,l,i,j) \in \mathcal{I}_{22}, \notag\\[1mm]
\hat{\mu}_{t-1}^{(r,l,i,j)} &= 
-\kappa_{t-1} 
+ \frac{(-1)^r |w_{t-1}^{(l,j)}| + (-1)^l |w_{t-1}^{(r,i)}|}
{|w_{t-1}^{(l,j)} - w_{t-1}^{(r,i)}|},
\quad (r,l,i,j) \in \mathcal{I}_{22}, \label{mu211}\\[1mm]
\hat{\mu}_{t-1}^{(r,i)} &= 
(-1)^r \kappa_{t-1} > -1,
\quad (r,i) \in \mathcal{I}_{21}. \label{mu221}
\end{align}

\begin{prop}\label{Prop-Conv21} 
Let \(g_t\) be a convex payoff function of the form \eqref{GeneralPayoffForm1}, and suppose that the AIP condition holds with 
\(\alpha_{t-1}^1 \leq 1-\kappa_{t-1} \).   
Then, the minimal superhedging price at time~\( t-1 \) is given by
\[
p_{t-1}(g_t)
= g_{t-1}(\phi_{t-2}, S_{t-1})
= \max_{(r,l,i,j)\in\mathcal{I}_{22}} g_{t-1}^{r,l,i,j}(\phi_{t-2}, S_{t-1})
\;\vee\;
\max_{(r,i)\in\mathcal{I}_{21}} g_{t-1}^{r,i}(\phi_{t-2}, S_{t-1}),
\]
where
\[
\begin{aligned}
g_{t-1}^{r,l,i,j}(\phi_{t-2}, x) 
&= \hat g_{t-1}^{r,l,i,j}(x) - \hat \mu_{t-1}^{(r,l,i,j)} \phi_{t-2} x, 
&\quad 
\hat g_{t-1}^{r,l,i,j}(x)
&= \lambda_{t-1}^{(r,l,i,j)} \tilde g_{t-1}^i(x)
  + \bigl(1 - \lambda_{t-1}^{(r,l,i,j)}\bigr) \tilde g_{t-1}^j(x), \\[1mm]
g_{t-1}^{r,i}(\phi_{t-2}, x) 
&= \hat g_{t-1}^{r,i}(x) - \hat \mu_{t-1}^{(r,i)} \phi_{t-2} x,
&\quad 
\hat g_{t-1}^{r,i}(x)
&= w_{t-1}^{(r,i)} \tilde g_{t-1}^1(x) 
  + \bigl(1 - w_{t-1}^{(r,i)}\bigr) \tilde g_{t-1}^i(x).
\end{aligned}
\]
\end{prop}

To express the optimal strategy associated to the minimal price above, we introduce the following quantities at time~$t-1$.
For \( i = 1, 2 \) and \( j = 2, \dots, 2N_t \), define
\[
\begin{aligned}
d_{t-1}^{(r,i)} &= \bigl(1 - \alpha_{t-1}^1 - (-1)^r \kappa_{t-1} \bigr) S_{t-1} b_{t-1}^i 
    + (-1)^r \kappa_{t-1} \phi_{t-2} S_{t-1},\\[1mm]
s_{t-1}^{(r,i)} &= 1 - \bigl(1 - \alpha_{t-1}^1 - (-1)^r \kappa_{t-1} \bigr) S_{t-1} p_{t-1}^i,\\[1mm]
A_{t-1}^{(r,i)}(v) &= d_{t-1}^{(r,i)} + s_{t-1}^{(r,i)} v,\\[1mm]
A_{t-1}^{(2)}(v) &= \max_{r=1,2}\max_{i=2,\dots,2N_t} A^{(r,i)}(v).
\end{aligned}
\]

The intersection points \( I_{r,l,i,j} \) are defined, when existence holds, by the equations
\[
A^{(r,i)}(I_{r,l,i,j}) = A^{(l,j)}(I_{r,l,i,j}), 
\quad r,l = 1,2,\; i,j = 2,\dots,2N_t.
\]\smallskip

The proof of the following result is given in Appendix~\ref{Pr-OptStrat21}, see Proof \ref{Pr-OptStrat21}.

\begin{prop}\label{OptStrat21}
Let \(g_t\) be a convex payoff function of the form \eqref{GeneralPayoffForm1}, and suppose that the AIP condition holds with 
\(\alpha_{t-1}^1 \leq 1-\kappa_{t-1} \). 
An optimal strategy associated with the minimal superhedging price \( p(g_t) \) is given by
\[
\phi^{\mathrm{opt}}_{t-1} = a_{t-1}(v^*),
\]
where 
\bean
&v^* \in \arg\min_{e \in I_{t-1}} \big| A_{t-1}^{(2)}(e) - p_{t-1}(g_t)+\hat g^1(\alpha_{t-1}S_{t-1}) \big|,\\
&I_{t-1} = \bigl\{ I_{r,l,i,j} : (r,i) \ne (l,j),\;
s_{t-1}^{(r,i)} \le 0,\;
s_{t-1}^{(l,j)} \ge 0 \bigr\}.
\eean
\end{prop}

Together, Propositions~\ref{Prop-Conv11}–\ref{Prop-Conv21} and Propositions~\ref{OptStrat11}–\ref{OptStrat21} provide a complete recursive characterization of the minimal superhedging prices and optimal hedging strategies under proportional transaction costs.
 \subsection{The multi-step super-hedging  under AIP}\label{MSP}

We now extend the one-step superhedging construction of Subsection \ref{GeneralStepProc1} to the whole interval $[0,T]$. This allows us to recursively compute minimal superhedging prices and associated optimal strategies under the AIP condition.  

Let $g_T$ be a non-negative convex payoff function. By Subsection \ref{GeneralStepProc1}, the minimal superhedging price at time $T-1$ is  
\[
p_{T-1}(g_T,\phi_{T-2}) = g_{T-1}(\phi_{T-2},S_{T-1}),
\]  
provided that the AIP condition holds, i.e.,  
\[
 \alpha_{T-1}^1  \leq 1+\kappa_{T-1} \quad \text{and} \quad \beta^{N_T}_{T-1}  \geq 1- \kappa_{T-1}. 
\]  
Since $N_T = 1$ and $\hat\mu_T^1 = 0$, the AIP condition at time $T-1$ does not depend on the terminal payoff $g_T$. Moreover, the backwardly computed payoff function $g_{T-1}$ exhibits the structural form (\ref{GeneralPayoffForm1}) with $N_{T-1} = 2$, where the coefficients $\hat\mu_{T-1}^i$, $i = 1, 2$, depend solely on the model parameters $\alpha_{T-1}$, $\beta_{T-1}$, and the transaction costs rate $\kappa_{T-1}$.

This reasoning naturally extends to earlier times by induction.
Suppose that, at some time $t$, a payoff function $g_t$ has been obtained from $g_T$ through the backward procedure described in Proposition~\ref{Prop-Conv11} or~\ref{Prop-Conv21}.
Then, the AIP condition imposed at the preceding step, between $t-1$ and $t$, depends solely on the model parameters $\alpha_{t-1}$, $\beta_{t-1}$, and $\kappa_{t-1}$, together with the structural coefficients of $g_t$, but not on the specific form of $g_T$.
In other words, the AIP condition can be regarded as a global property, independent of the terminal payoff. To formalize this, we introduce a set-valued backward sequence $(\Gamma_t)_{t=T,\dots,1}$:  
\begin{itemize}
\item The initial value is $\Gamma_T = \{\hat\mu_T^1\}$ with $\hat\mu_T^1 = 0$.  
\item Denote $N_t = {\rm card}(\Gamma_t)$ for $t = T,\dots,1$.
\item For $t \ge 1$, $\Gamma_{t-1}$ is defined as the set of strictly increasing coefficients $\hat \mu_{t-1}^i$ obtained recursively from the previous step:
\begin{itemize}
\item If $ 1-\kappa_{t-1}\le \alpha_{t-1}^1 \leq 1+\kappa_{t-1} $ (Case \ref{case:large} under AIP), the coefficients are given by (\ref{mu111}) and (\ref{mu121}).  
\item If $\alpha_{t-1}^1 \le 1-\kappa_{t-1} $ and $\beta_{t-1}^{N_t}\geq 1-\kappa_{t-1}$ (Case \ref{case:mild} under AIP), the coefficients are given by (\ref{mu211}) and (\ref{mu221}).  
\item Otherwise, set $\Gamma_{t-1} = \emptyset$.
\end{itemize}

\end{itemize}

Intuitively, $\Gamma_t$ encapsulates all strictly increasing coefficients required at each step to compute the minimal superhedging price.
This sequence depends solely on the model parameters and guarantees the absence of immediate profit (AIP) at all times, independently of the terminal payoff $g_T$.
Consequently, the backward induction can be carried out iteratively from $t = T$ down to $t = 0$, yielding the multi-step minimal superhedging price $g_0(\phi_{-1}, S_0)$ with $\phi_{-1} = 0$.

\begin{defi}\label{DefAIPGlob}  
The AIP condition is said to hold over the interval $[0,T]$ if, for all $t = 1, \dots, T$,
\[
 \alpha_{t-1}^1 \le 1+\kappa_{t-1} \quad \text{and} \quad \beta_{t-1}^{N_t} \ge 1-\kappa_{t-1}, \quad \text{with } N_t \neq 0,
\]  
where the coefficients are defined by
\begin{align*}
\alpha_{t-1}^1 &:= \alpha_{t-1}(1 + \hat\mu_t^1), & \beta_{t-1}^{N_t} &:= \beta_{t-1}(1 + \hat\mu_t^{N_t}), \\
\hat\mu_t^1 &:= \min \Gamma_t, & \hat\mu_t^{N_t} &:= \max \Gamma_t.
\end{align*}

\end{defi}

Under this condition, the minimal superhedging price of the zero payoff is trivially equal to zero at any time $t-1 \ge 0$ whenever $\phi_{t-2} = 0$. The AIP condition guarantees that the backward procedure is well-defined at each step and rules out situations in which the minimal price would become infinitely negative, which are economically meaningless.  

In the absence of transaction costs ($\kappa_t = 0$ for all $t$), the AIP condition simplifies to  
\[
\alpha_t \le 1 \le \beta_t, \quad t = 0, \dots, T,
\]  
which is observed in practice. Conversely, the condition of Definition \ref{DefAIPGlob} shows that arbitrage opportunities in a frictionless model can be eliminated by suitably increasing transaction costs.  These considerations lead to the following main result:  

\begin{theo} \label{BS} 
Let $g_T$ be a non-negative convex payoff function. Suppose that (\ref{MS}) holds  and the deterministic transaction costs coefficients $\kappa_t \in [0,1)$ are such that  the global AIP condition holds. Then, at time $0$, there exists a minimal superhedging price  $
p_0 = g_0(\phi_{-1}, S_0), \quad \phi_{-1} = 0$,
where $(g_t)_{t=0,\dots,T}$ satisfy the terminal condition $g_T(\phi_{T-1}, s) = g_T(s)$ and are defined backward by Propositions \ref{Prop-Conv11} or \ref{Prop-Conv21}. Moreover, the associated optimal strategy $(\phi_t)_{t=-1,0,\dots,T-1}$ satisfies $\phi_{-1}=0$ and is constructed forward via Propositions \ref{OptStrat11} or \ref{OptStrat21}.  
\end{theo}

\begin{rem}
The above Theorem \ref{BS} provides a practical backward-forward procedure to compute the minimal superhedging price and the associated optimal strategy. Let $g_t(\phi_{t-1}, S_t)$ denote the minimal price at time $t$. The procedure may be implemented as follows:

\begin{enumerate}
    \item \textbf{Backward step:} Starting from $g_T = g_T(S_T)$, the payoff functions $g_{t-1}$ are computed recursively from $g_t$ for $t = T, T-1, \dots, 1$ using Proposition~\ref{Prop-Conv11} or~\ref{Prop-Conv21}.
Each $g_{t-1}$ is given by the maximum over convex combinations of the time-$t$ functions defining $g_t$.
At the end of this backward procedure, the initial pricing function $g_0$ is obtained.
    \item \textbf{Forward step:} Initialize $\phi_{-1} = 0$ and iteratively compute the optimal strategy $\phi_t$ for $t=0, \dots, T-1$ using Proposition~\ref{OptStrat11} or \ref{OptStrat21}. 
\end{enumerate}
\end{rem}

% ===== Section 3 Numerical Implementation =====
\section{Numerical implementation}\label{NI}

We illustrate the method developed above using the U.S. stock \texttt{SPY} (precisely a tracker S\&P 500 ETF for us) as the underlying asset $S$, and we consider a European call option with payoff $g(S_T) = (S_T - K)^+$.
The data consist of daily \texttt{SPY} prices from June~3,~2013 to December~30,~2016, yielding a total of 904 observations and 903 return observations. Prices of the tracker ranged from $126.99$ to $196.60$. In this example, we assume that a week consists of the first four trading days, from Monday to Thursday.
A calibration window $W^j$ of $k = 1, \dots, 52$ weeks is used to estimate the conditional support of relative prices for the following weeks, indexed by $j = 1, \dots, 100$.
Specifically, we assume that the price process $S^{(j)}$ for the $j$-th week satisfies
\[
\operatorname{supp}_{\mathcal{F}_t}\!\Big(\frac{S^{(j)}_{t+1}}{S^{(j)}_t}\Big) = [\alpha^{(j)}_t,\, \beta^{(j)}_t],
\]
where
\[
\alpha^{(j)}_t = \min_{k \in \text{$W^j$}} \frac{S^{(k)}_{t+1}}{S^{(k)}_t}, 
\qquad 
\beta^{(j)}_t = \max_{k \in \text{$W^j$}} \frac{S^{(k)}_{t+1}}{S^{(k)}_t}.
\]

 All Call options are supposed to be at-the-money (ATM), i.e. $K=S_0$. We consider the  proportional transaction costs rates $\kappa=0.2j$ for $j=1,\cdots,10$.

\begin{figure}[htbp]
    \centering
    \includegraphics[width=1\textwidth]{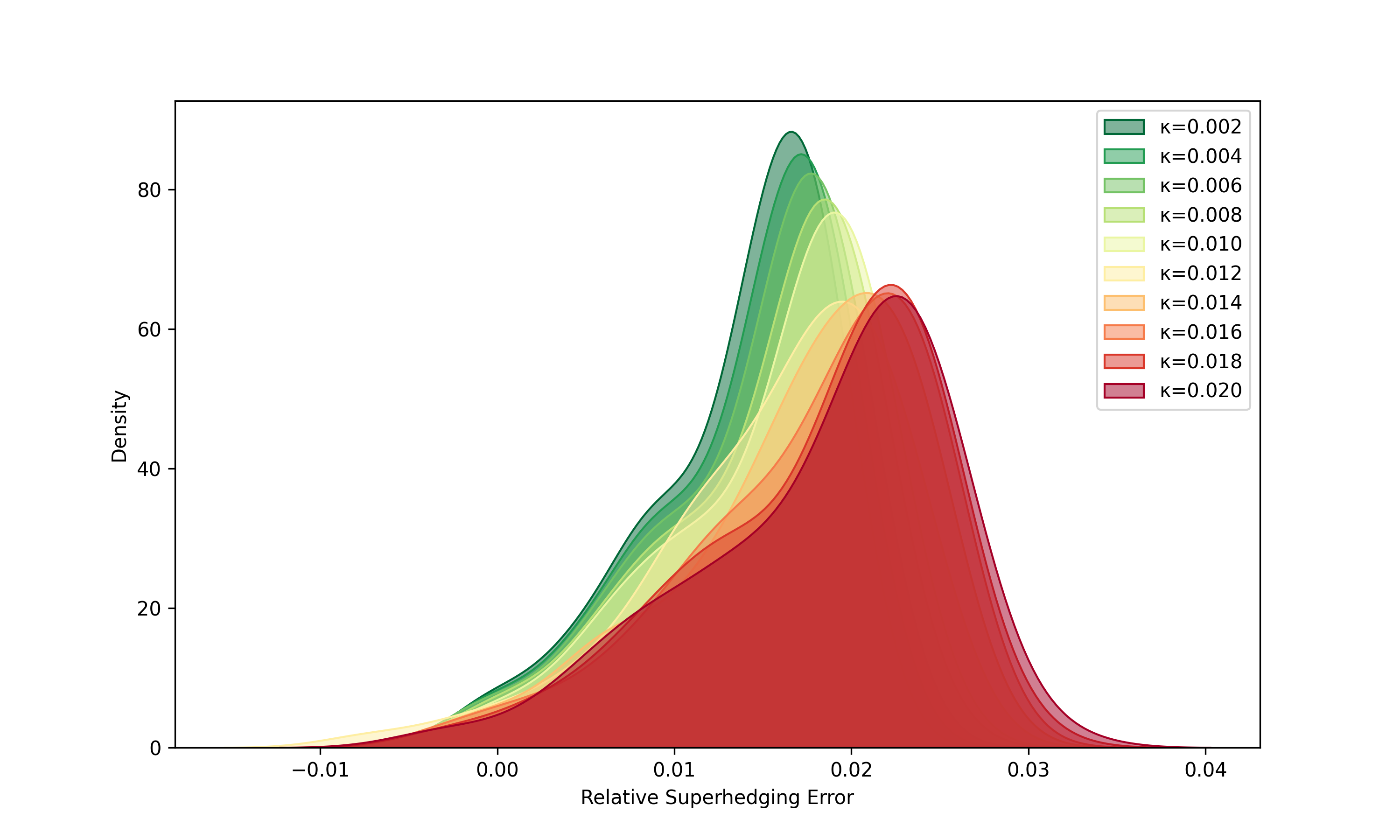}
    \caption{Distributions of the relative superhedging error $(V_T -(S_T - K)^+) / S_0$ for various transaction costs rates.}
    \label{fig:kde}
\end{figure}

\begin{figure}[htbp]
    \centering
    \includegraphics[width=1\textwidth]{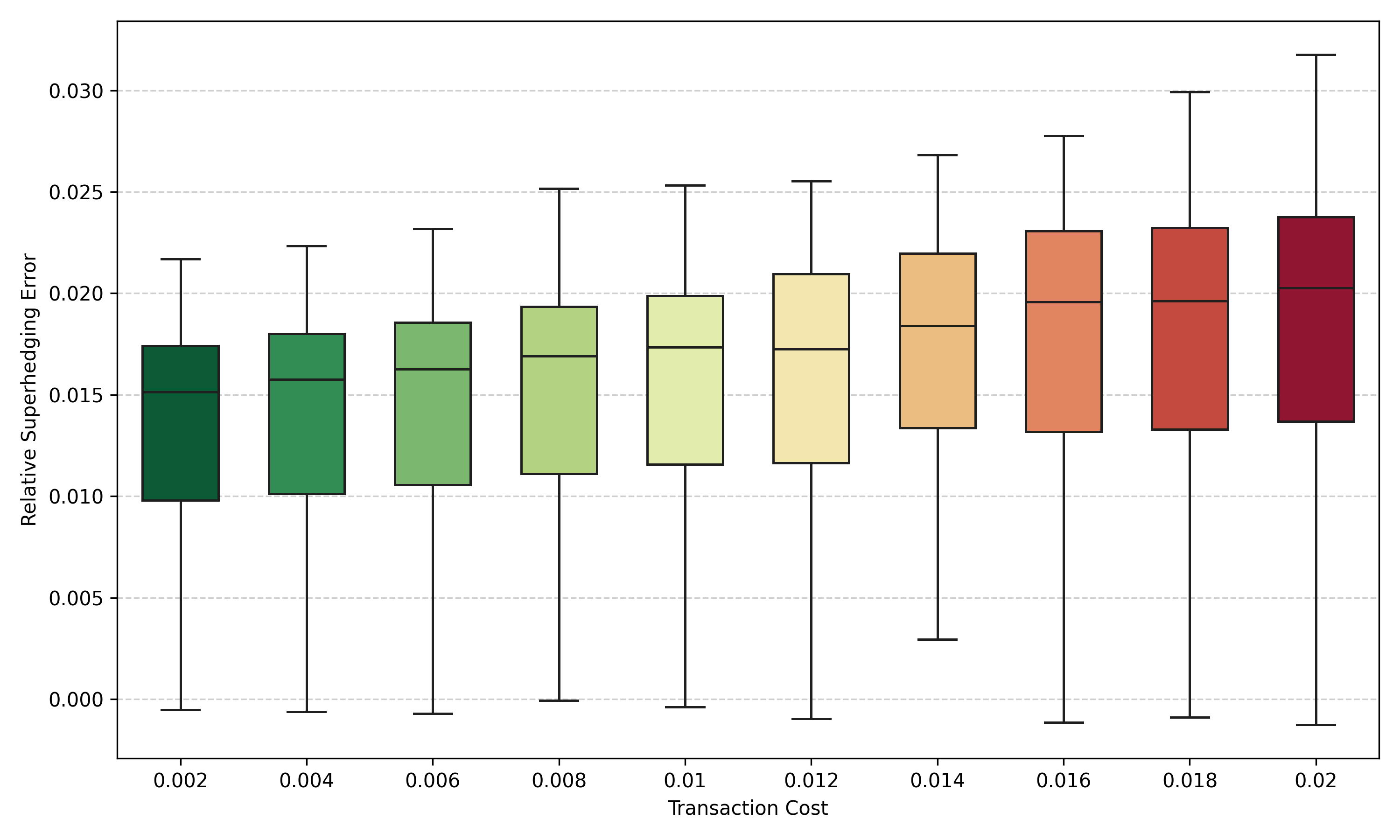}
    \caption{Boxplots of the relative superhedging errors for various transaction costs rates.}
    \label{fig:boxplot}
\end{figure}

\begin{figure}[htbp]
    \centering
    \includegraphics[width=0.8\textwidth]{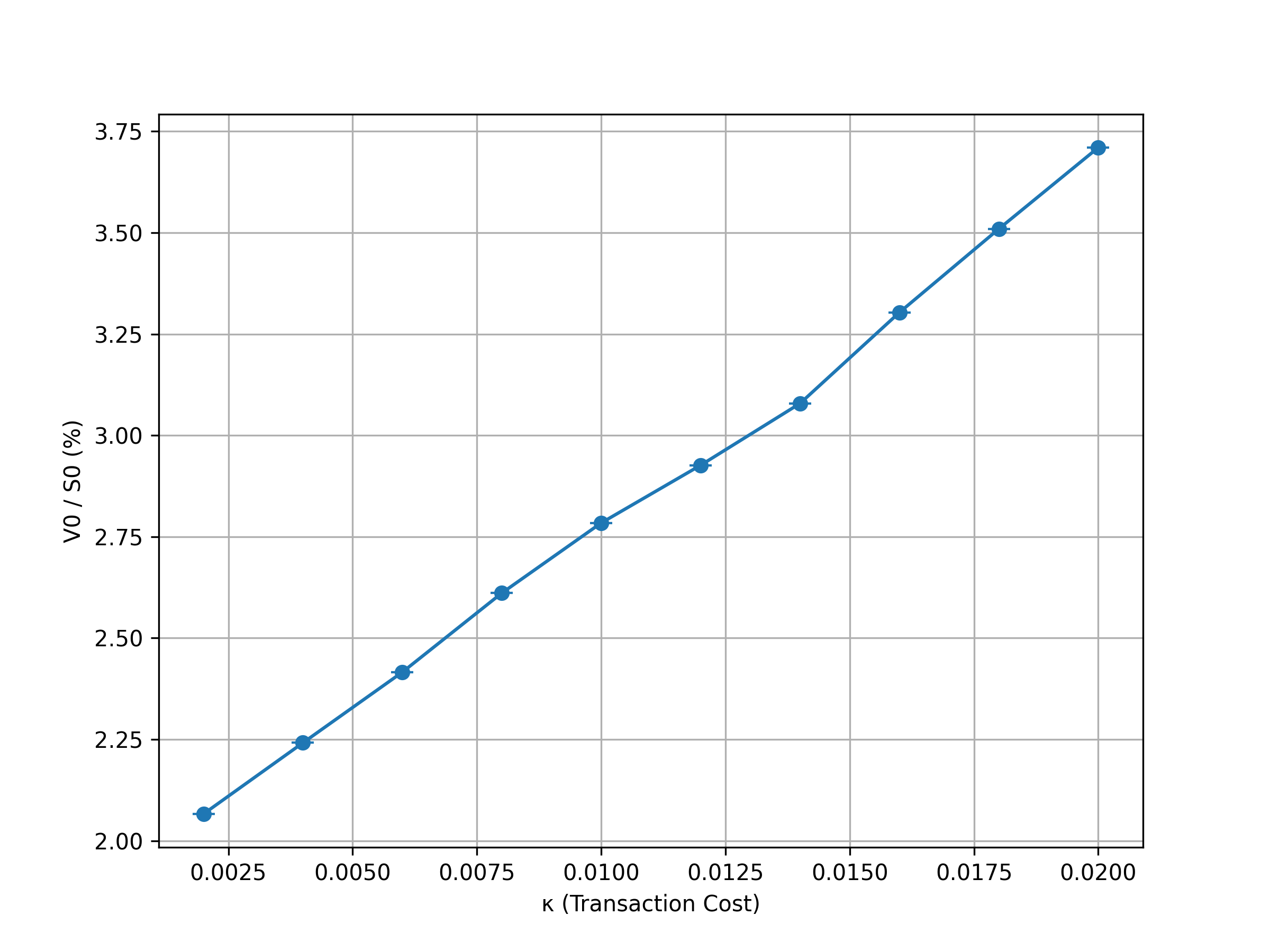}
    \caption{Impact of transaction costs on the initial relative superhedging price (in percent).}
    \label{fig:transactionVsCost}
\end{figure}

\begin{table}[htbp]
\centering
\caption{Superhedging performance statistics for SPY over 100 weeks for various proportional transaction costs rates ($\varepsilon$ is the relative superhedging error, $V_0$ the initial price).}
\label{tab:performance_stats_spy}

\sisetup{
    table-number-alignment = center,
    round-mode = places,
    round-precision = 2,
    detect-weight = true,
    detect-family = true
}

\begin{tabular}{S S S S S }
\toprule
{\(\kappa\) (\%)} & {Mean \(\bar{\varepsilon}\) (\%)} & {Std. Dev. \(\sigma_\varepsilon\) (\%)}  & {\(V_0/S_0\) (\%)} & {\(\mathbb{P}(\varepsilon \ge 0)\)} \\
\midrule
0.2  & 1.36 & 0.53 & 2.07 & 0.97 \\
0.4  & 1.41 & 0.55 & 2.24 & 0.98 \\
0.6  & 1.45 & 0.57 & 2.42 & 0.98 \\
0.8  & 1.52 & 0.59 & 2.61 & 0.98 \\
1.0  & 1.56 & 0.61 & 2.78 & 0.98 \\
1.2  & 1.58 & 0.67 & 2.93 & 0.97 \\
1.4  & 1.71 & 0.66 & 3.08 & 0.98 \\
1.6  & 1.77 & 0.68 & 3.30 & 0.98 \\
1.8  & 1.80 & 0.69 & 3.51 & 0.98 \\
2.0  & 1.84 & 0.71 & 3.71 & 0.98 \\
\bottomrule
\end{tabular}
\end{table}

\smallskip

As illustrated in Figures~\ref{fig:kde} and~\ref{fig:boxplot}, most of the relative superhedging errors observed during the $100$ weeks are not negative.
This is expected when the calibration of $\alpha_t^{(j)}$ and $\beta_t^{(j)}$ based on past data is consistent with future outcomes.
Conversely, if the calibration is inaccurate, negative superhedging errors may appear, and they tend to be larger as the transaction costs rate increases.

As shown in Figure~\ref{fig:transactionVsCost}, the relative initial price increases with the transaction costs rate, ranging from 2.07\% to 3.71\% of $S_0$, as an affine function. Precisely, from our experiment, we have  $V_0\simeq (0.9111\kappa+0.0188)S_0,\, \kappa\in [0.2\%,2\%].$ Table~\ref{tab:performance_stats_spy} summarizes these statistics, confirming that our \textbf{ method  provides reliable hedging with limited relative error}, even in the presence of substantial transaction costs.
Approximately 9.7\% of the realized prices fall outside the estimated support interval $[\alpha_{t}^{(j)} S_{t}^{(j)},\, \beta_{t}^{(j)} S_{t}^{(j)}]$, which explains the observed 2--3\% of negative errors. Note that the AIP condition is satisfied for every week considered, ensuring the absence of arbitrage in the model.

% ===== Section 4 =====
\section{The general one step super-hedging problem}\smallskip \label{section propagation}

At time $t$, we consider a payoff function $g_t$ depending on the previous strategy $\phi_{t-1} \in L^0(\mathbb{R}, \mathcal{F}_{t-1})$, given by \eqref{GeneralPayoffForm1} and satisfying \eqref{formGi-11}–\eqref{formGi1}. Our goal is to determine the superhedging price $V_{t-1}$ such that $V_t\geq g_t(\phi_{t-1},S_t)$ for some strategy $\phi_{t-1}$. This  is equivalent to 
\bea \label{replication1}
V_{t-1}&\geq& g_t(\phi_{t-1},S_t)-\phi_{t-1}\Delta S_{t}+\kappa_{t-1}| \phi_{t-1}-\phi_{t-2}|S_{t-1},\\ \nonumber
&\geq& \esssup_{\cF_{t-1}}(g_t(\phi_{t-1},S_t)-\phi_{t-1}S_t)+\phi_{t-1}S_{t-1}+\kappa_{t-1}|\Delta\phi_{t-1}|S_{t-1}.
\eea 
Let us introduce the following processes
\bea
\gamma_t^i(\phi_{t-1},x)&:=& g_t^i(\phi_{t-1},x)-\phi_{t-1}x,\quad i=1,\cdots,N_t,\\
\gamma_{t}&:=&\max_{i=1,\dots,N_t}\gamma_t^i\\
\bar g_t^i(x)&=&\hat g_t^i(\frac{x}{1+\hat\mu^i_t}),\\
\bar f_{t-1}^i &:=& -\bar g_t^i+\delta_{K_{t-1}^i},\quad K_{t-1}^i=(1+\hat \mu_{t}^i)C_{t-1},\\ \label{ftmoinsun}
f_{t-1}&:=&\min_{i=1,\dots,N_t}\bar f_{t-1}^i,
\eea
where $C_{t-1}={\rm supp}_{\cF_{t-1}}(S_t)=[\alpha_{t-1}S_{t-1},\beta_{t-1}S_{t-1}]$.
\begin{prop} We have:
    \bean
    \esssup_{\cF_{t-1}}(g_t(\phi_{t-1},S_t)-\phi_{t-1}S_t)=f_{t-1}^*(-\phi_{t-1})
    \eean
\end{prop}
\begin{proof}
By definition, we have:
    \bean
    \esssup_{\cF_{t-1}}(g_t(\phi_{t-1},S_t)-\phi_{t-1}S_t)
    &=&\esssup_{\cF_{t-1}}\gamma_{t}(\phi_{t-1},S_t),\\
    &=& \esssup_{\cF_{t-1}}(\max_{i=1,\dots,N_t}\gamma_t^i(\phi_{t-1},S_t),\\
    &=& \max_{i=1,\dots,N_t}\esssup_{\cF_{t-1}}\gamma_t^i(\phi_{t-1},S_t).
    \eean
    Since the functions $g_t^i$, $i = 1, \dots, N_t$, are continuous, we deduce from Proposition~\ref{esssup-sup} that 
    \bean
   \esssup_{\cF_{t-1}}\gamma_t^i(\phi_{t-1},S_t)&=& \sup_{z\in C_{t-1}^i}(\hat g_t^i(z)-\hat \mu_t^i\phi_{t-1}z-\phi_{t-1}z)\\
   &=& \sup_{z\in C_{t-1}^i}(\hat g_t^i(z)-(1+\hat \mu_t^i)\phi_{t-1}z)\\
   &=& \sup_{x\in \kappa_{t-1}^i}(\bar g_t^i(x)-\phi_{t-1}x),\quad x=(1+\hat \mu_t^i)z,\\
   &=& \sup_{x\in \mathbb{R}}(-\phi_{t-1}x-\bar f_{t-1}^i(x))\\
   &=& (\bar f_{t-1}^i)^*(-\phi_{t-1}).
    \eean 
    Therefore, 
    \bean \esssup_{\cF_{t-1}}\gamma_{t}(\phi_{t-1},S_t)&=&\max_{i=1,\dots,N_t}(\bar f_{t-1}^i)^*(-\phi_{t-1}),\\
    &=& \left(\min_{i=1,\dots,N_t}\bar f_{t-1}^i\right)^*(-\phi_{t-1})= f_{t-1}^*(-\phi_{t-1}).\eean
    The conclusion follows.
\end{proof}
Finally, $V_{t-1}$ is a superhedging price if and only if \eqref{replication1} holds, 
or equivalently,
\[
V_{t-1} \geq P_{t-1}(\phi_{t-1})
:= f_{t-1}^*(-\phi_{t-1})
- \phi_{t-1} S_t
+ \phi_{t-1} S_{t-1}
+ \kappa_{t-1} |\Delta \phi_{t-1}| S_{t-1}.
\]
Thus, the set of all superhedging prices at time $t-1$ is given by
\bea \label{setsuperhedgingprices}
\mathcal{P}_{t-1}(g_t)
= \big\{ P_{t-1}(\phi_{t-1}) : \phi_{t-1} \in L^0(\mathbb{R}, \mathcal{F}_{t-1}) \big\}
+ L^0(\mathbb{R}_+, \mathcal{F}_{t-1}).
\eea
The infimum superhedging price is defined as
\bea
p_{t-1}(g_t)
= \essinf \, \mathcal{P}_{t-1}(g_t).
\eea

Throughout the sequel, we denote by $\Phi_{\phi_{t-2}}$ the invertible mapping defined by
\bea \label{Phi}
\Phi_{\phi_{t-2}}(x) := x - \kappa_{t-1} \, |x + \phi_{t-2}|.
\eea
Furthermore, we define
\bea \label{hatPhi}
\hat{\Phi}_{\phi_{t-2}}(x) := - \Phi_{\phi_{t-2}}(-x).
\eea

\begin{theo}\label{infPrice}
    Let $g_t$ be a convex payoff function of the form \eqref{GeneralPayoffForm1}.  
Then the infimum superhedging price at time $t-1$ satisfies
\bea
p_{t-1}(g_t)
= - \big( f_{t-1}^* \circ \Phi_{\phi_{t-2}}^{-1} \big)^*(S_{t-1}),
\eea
where $f_{t-1}$ is defined by~\eqref{ftmoinsun} and $\Phi_{\phi_{t-2}}^{-1}$ is the inverse function of $\Phi_{\phi_{t-2}}$ defined by \eqref{Phi}.
\end{theo}
\begin{proof}
    Recall that the infimum price is defined as $p_{t-1}(g_t):=\essinf \cP_{t-1}(g_t)$. By  Proposition \ref{PropEssinfNI},  we get the following
\bean
p_{t-1}(g_T)&=&\inf_{z\in \R}\left(f_{t-1}^*(-z)+z S_{t-1}+\kappa_{t-1}|z-\phi_{t-2}|S_{t-1}\right),\\
&=&-\sup_{z\in \R }\left(-f_{t-1}^*(-z)-z S_{t-1}-\kappa_{t-1}|z-\phi_{t-2}|S_{t-1}\right),\\
&=&-\sup_{z\in \R }\left(-f_{t-1}^*(z)+z S_{t-1}-\kappa_{t-1}|z+\phi_{t-2}|S_{t-1}\right), \\
&=&-\sup_{z\in \R}\left(S_{t-1}\Phi_{\phi_{t-2}}(z)-f_{t-1}^*(z) \right),\\
&=&-\sup_{z\in\R}\left(S_{t-1}z-f_{t-1}^* \circ \Phi_{\phi_{t-2}}^{-1}(z)\right),\\
&=&-(f_{t-1}^{*} \circ \Phi_{\phi_{t-2}}^{-1})^{*}(S_{t-1}).
\eean

\end{proof}

    \begin{theo}
       Let $g_t$ be a convex payoff function of the form \eqref{GeneralPayoffForm1}.  
There exists a unique $\mathcal{F}_{t-1}$-measurable convex integrand $h_{t-1}$ such that
\[
p_{t-1}(g_t) = -\, h_{t-1}(S_{t-1}).
\]
Moreover,
\bea \label{h_{t-1}}
h_{t-1}
= \big[ \big( \bar h_{t-1}^* \circ \Phi_{\phi_{t-2}} \big)^{**} 
\circ \Phi_{\phi_{t-2}}^{-1} \big]^*,
\eea
where
\bea \label{barh_{t-1}}
\bar h_{t-1}(x)
:= \sup \big\{
  ( \gamma^* \circ \Phi_{\phi_{t-2}}^{-1} )^*(x)
  : \gamma \in \mathcal{A},\, \gamma \le f_{t-1}
\big\}.
\eea
\end{theo}
\begin{proof}
By Theorem~\ref{infPrice}, we have
\[
p_{t-1}(g_t)
= - \big( f_{t-1}^{*} \circ \Phi_{\phi_{t-2}}^{-1} \big)^{*}(S_{t-1}).
\]
Since $f_{t-1} = +\infty$ on $\mathbb{R}_{-}$ and $\Phi_{\phi_{t-2}}$ is a bijection such that 
$\Phi_{\phi_{t-2}}^{-1}$ is convex, Proposition~\ref{propconjf} applies. 
Hence, there exists a unique lower semicontinuous convex function $h_{t-1}$ such that
\[
f_{t-1}^* = h_{t-1}^* \circ \Phi_{\phi_{t-2}}.
\]
Therefore,
\bean
p_{t-1}(g_t)
&=& - \big( f_{t-1}^{*} \circ \Phi_{\phi_{t-2}}^{-1} \big)^{*}(S_{t-1}) \\[0.3em]
&=& - \big( h_{t-1}^* \circ \Phi_{\phi_{t-2}} \circ \Phi_{\phi_{t-2}}^{-1} \big)^{*}(S_{t-1}) \\[0.3em]
&=& -\, h_{t-1}(S_{t-1}). 
\eean
\end{proof}

\subsection{Computation of $\bar{h}_{t-1}$}\smallskip

In the following, $\gamma_{a,b} \in \cA$ denotes any affine function 
$\gamma_{a,b}(x) = a x + b$. 
Recall that $\bar h_{t-1}$ is given by~\eqref{barh_{t-1}}. 
Since $\Phi_{\phi_{t-2}}^{-1}$ is bijective, Lemma~\ref{LemmDualComp} yields
\[
(\gamma_{a,b}^* \circ \Phi_{\phi_{t-2}}^{-1})^*(x)
= \Phi_{\phi_{t-2}}(a)x + b.
\]
Therefore,
\begin{align}
\bar h_{t-1}(x)
&= \sup_{a,b \in \mathbb{R}} \big\{ \Phi_{\phi_{t-2}}(a)x + b,\; \gamma_{a,b} \le f_{t-1} \big\} \label{barh}\\
&= \sup_{a,b \in \mathbb{R}} \big\{ \Phi_{\phi_{t-2}}(a)x + b,\; \gamma_{a,b} \le \min_{1,\dots,N_t}(\bar f_{t-1}^i) \big\} \notag\\
&= \sup_{a,b \in \mathbb{R}} \big\{ \Phi_{\phi_{t-2}}(a)x + b,\; 
\gamma_{a,b} \le -\bar g_t^i + \delta_{K_{t-1}^i},\; i=1,\dots,N_t \big\} \notag\\
&= -\inf_{a,b \in \mathbb{R}} 
\big\{ -\Phi_{\phi_{t-2}}(-a)x + b,\; 
\gamma_{a,b} \ge \bar g_t^i \text{ on } K_{t-1}^i,\; i=1,\dots,N_t \big\}. \notag
\end{align}

\medskip
In the sequel, assume that $C_{t-1} = [m_{t-1}, M_{t-1}]$ is the conditional support of $S_t$ given $\cF_{t-1}$, 
where $m_{t-1}$ and $M_{t-1}$ are $\cF_{t-1}$-measurable with $m_{t-1} < M_{t-1}$ a.s. 
In fact, by assumption, $m_{t-1} = \alpha_{t-1} S_{t-1}$ and $M_{t-1} = \beta_{t-1} S_{t-1}$.

For each $j = 1, \dots, N_t$, define
\begin{align}
m_{t-1}^j &= (1 + \hat{\mu}_{t}^j) m_{t-1}, &
m_{t-1}^{j,+} &= \frac{m_{t-1}^j}{1 + \kappa_{t-1}}, &
m_{t-1}^{j,-} &= \frac{m_{t-1}^j}{1 - \kappa_{t-1}}, \notag\\
M_{t-1}^j &= (1 + \hat{\mu}_{t}^j) M_{t-1}, &
M_{t-1}^{j,+} &= \frac{M_{t-1}^j}{1 + \kappa_{t-1}}, &
M_{t-1}^{j,-} &= \frac{M_{t-1}^j}{1 - \kappa_{t-1}}, \notag\\
\bar{y}_{t-1}^j &= \bar{g}_{t}^j(m_{t-1}^j), &
\bar{Y}_{t-1}^j &= \bar{g}_{t}^j(M_{t-1}^j). \notag
\end{align}
Observe that $\bar{y}_{t-1}^j = \tilde g_{t-1}^j(S_{t-1})$ and 
$\bar{Y}_{t-1}^j = \tilde g_{t-1}^{j+N_t}(S_{t-1})$, 
where $\tilde g_{t-1}$ is defined in \eqref{tilde-g}.

We then define the following slopes and intercepts by
\begin{align}
\bar p_{t-1}^{j} &= \frac{1}{m_{t-1}^{j} - m_{t-1}^1}, &
\bar b_{t-1}^{j} &= \bar p_{t-1}^{j} \big( \bar{y}_{t-1}^{j} - \bar{y}_{t-1}^1 \big),
& j = 2, \dots, N_t, \notag\\
\bar p_{t-1}^{j} &= \frac{1}{M_{t-1}^{j-N_t} - m_{t-1}^1}, &
\bar b_{t-1}^{j} &= \bar p_{t-1}^{j} \big( \bar{Y}_{t-1}^{j-N_t} - \bar{y}_{t-1}^1 \big),
& j = N_t+1, \dots, 2N_t. \label{pente-pi}
\end{align}

and we set $\bar p_{t-1}^{j} = 0$ and $\bar b_{t-1}^{j} = -\infty$ 
whenever $m_{t-1}^{j} = m_{t-1}^1$ or $M_{t-1}^{j} = m_{t-1}^1$.

Let $\sigma$ be the permutation of $\{2, \dots, 2N_t\}$ that arranges $(\bar p_{t-1}^i)_{i=2}^{2N_t}$ in nondecreasing order, 
so that $p_{t-1}^i := \bar p_{t-1}^{\sigma(i)}$ with $p_{t-1}^2 \le \dots \le p_{t-1}^{2N_t}$. 
The same permutation $\sigma$ is applied to $(\bar b_{t-1}^i)_{i=2}^{2N_t}$, yielding $b_{t-1}^i := \bar b_{t-1}^{\sigma(i)}$ for $i = 2, \dots, 2N_t$.

We now reintroduce the mapping $a_{t-1}$ defined in \eqref{Deltahedging} in Section~\ref{MR},
\[
a_{t-1} := \max_{i = 2, \dots, 2N_t} \big( b_{t-1}^i - p_{t-1}^i \alpha \big),
\]

\begin{rem}
By convention, $b_{t-1}^i = -\infty$ whenever $p_{t-1}^i = 0$. 
Hence, the slopes $-p_{t-1}^i$ of the affine functions
\[
\alpha \mapsto b_{t-1}^i - p_{t-1}^i \alpha
\]
defining $a_{t-1}$ are all strictly negative. Consequently, 
$a_{t-1}$ is a strictly decreasing and continuous function and, therefore, invertible.
\end{rem}

\begin{lemm}\label{hbar} We have $\bar{h}_{t-1}(x)=-\bar\varphi_{t-1}(x)$ where 
\bea
\bar\varphi_{t-1}(x) &=& \inf_{\alpha \geq 0, \, a \geq a_{t-1}(\alpha)} \varphi_{t-1}(\alpha,a,x),\\
\label{varphi}
\varphi_{t-1}(\alpha,a,x) &= &\hat\Phi_{\phi_{t-2}}(a)x + \bar y_{t-1}^1 - a m_{t-1}^1 + \alpha.
\eea
\end{lemm}
\begin{proof} 
By (\ref{barh}), we  solve the inequality $\gamma_{a,b}\geq \bar g_{t-1}^i$ on the set $K_{t-1}^i=[m_{t-1}^i,M_{t-1}^i]$. We rewrite $\gamma_{a,b}(x)=\bar y_{t-1}^i+\alpha_i +a(x-m_{t-1}^i)$ and we impose $\alpha_i\ge 0$ since we want $\gamma_{a,b}(m_{t-1}^i)\geq \bar g_{t-1}^i(m_{t-1}^i)=\bar y_{t-1}^i$. We define $\alpha=\alpha_1$.  Since $b=\bar y_{t-1}^i+\alpha_i-am_{t-1}^i$, for any $i=1,\cdots,N_t$, the constraint $\alpha_i\ge 0$ for $i\ge 2$ reads as $b-\bar y_{t-1}^i+am_{t-1}^i\ge 0$, with $b=\bar y_{t-1}^1+\alpha-am_{t-1}^1$, i.e. $a\ge  b^i_{t-1}- p_{t-1}^i\alpha$.

Since $\bar g_{t-1}^i$ is convex, $\gamma_{a,b}\geq \bar g_{t-1}^i$ on $K_{t-1}^i=[m_{t-1}^i,M_{t-1}^i]$ if and only if  $\gamma_{a,b}(M_{t-1}^i)\geq \bar g_{t-1}^i(M_{t-1}^i)=\bar Y_{t-1}^i$, i.e. $aM_{t-1}^i+\bar y_{t-1}^1+\alpha-a m_{t-1}^1\ge \bar Y_{t-1}^i$ or equivalently $a\ge  b^i_{t-1}- p_{t-1}^i\alpha$ for $i=N_t,\cdots,2N_t-1$.\smallskip

By \eqref{barh}, we deduce that $\bar{h}_{t-1}(x)=-\bar\varphi_{t-1}(x)$ where
 
\bean
\bar\varphi_{t-1}(x)=\inf_{\alpha \geq 0, \, a \geq a_{t-1}(\alpha)} \left\{ \hat\Phi_{\phi_{t-2}}(a)x + \bar y_{t-1}^1 - a m_{t-1}^1 + \alpha \right\},
\eean
where $\hat\Phi_{\phi_{t-2}}(a) = -\Phi_{\phi_{t-2}}(-a)$, i.e. $
\bar\varphi_{t-1}(x) = \inf_{\alpha \geq 0, \, a \geq a_{t-1}(\alpha)} \varphi_{t-1}(\alpha,a,x),$
with  
\bean
\varphi_{t-1}(\alpha,a,x) = \hat\Phi_{\phi_{t-2}}(a)x + \bar y_{t-1}^1 - a m_{t-1}^1 + \alpha.
\eean 
\end{proof}

\begin{theo}\label{optimization}
We have 
    \bean 
\bar \varphi_{t-1}(x)&=&-\infty,\quad x<m_{t-1}^{1+}\,{\rm or }\, x>M_{t-1}^{N-},\\
%&=&\inf_{0\leq \alpha \leq a_{t-1}^{-1}(\phi_{t-2})\vee 0}\varphi_{t-1}(\alpha,a_{t-1}(\alpha)\vee\phi_{t-2},x),\quad  x\in[m_{t-1}^{1+},m_{t-1}^{1-}],\\
&=&\inf_{\alpha\geq 0}\varphi_{t-1}(\alpha,a_{t-1}(\alpha)\vee\phi_{t-2},x),\quad  x\in[m_{t-1}^{1+},m_{t-1}^{1-}],\\
&=&\inf_{\alpha \ge 0}\varphi_{t-1}(\alpha,a_{t-1}(\alpha),x),\quad  x\in[m_{t-1}^{1-},M_{t-1}^{N-}],\\
\eean
\end{theo}
\begin{proof}
    Recall that, the mapping  $\varphi_{t-1}$ is given by the following,
\bean
\varphi_{t-1}(\alpha,a,x) &=&\hat\Phi_{\phi_{t-2}}(a)x + \bar y_{t-1}^1-am_{t-1}^1+\alpha\\\notag
&=& (a+\kappa_{t-1}|a-\phi_{t-2}|)x+\bar y_{t-1}^1-am_{t-1}^1+\alpha\\
&=& \left\{
    \begin{array}{ll}
        \eta_{t-1}^1(x)a+\kappa_{t-1}\phi_{t-2}x+\bar y_{t-1}^1 +\alpha :=\varphi_{t-1}^1(\alpha,a,x)& \mbox{if } a\leq \phi_{t-2} \\
        \eta_{t-1}^2(x)a-\kappa_{t-1}\phi_{t-2}x+\bar y_{t-1}^1 +\alpha:= \varphi_{t-1}^2(\alpha,a,x) & \mbox{if } a\geq \phi_{t-2} \notag
    \end{array}
\right.
\eean
where, for $j=1,2$, $\eta_{t-1}^j(x):=\left((1+(-1)^j \kappa_{t-1})x-m_{t-1}^1\right)$.

First, consider $x<m_{t-1}^{1+}$ so that $\eta_{t-1}^1(x)$ and $\eta_{t-1}^2(x)$ are negative. Thus, by taking $a=\infty$,\bea\label{case1} \bar\varphi_{t-1}(x)=\inf_{\alpha \geq 0, \, a \geq a_{t-1}(\alpha)}\varphi_{t-1}(\alpha,a,x)=-\infty.\eea

Next, if $x\in[m_{t-1}^{1+},m_{t-1}^{1-}]$, then $\eta_{t-1}^1(x)\leq 0$ and $\eta_{t-1}^2(x)\geq 0$. Consequently, the mappings $a\to \varphi_{t-1}^1(\alpha,a,x)$ and $a\to \varphi_{t-1}^2(\alpha,a,x)$ are respectively non-increasing and  non-decreasing. This implies that \bea\label{case2} \bar\varphi_{t-1}(x)=\inf_{\alpha \geq 0, \, a \geq a_{t-1}(\alpha)}\varphi_{t-1}(\alpha,a,x)=\inf_{\alpha\geq 0}\varphi_{t-1}(\alpha,a_{t-1}(\alpha)\vee\phi_{t-2},x).\eea

If $a_{t-1}(0)\leq \phi_{t-2}$, then $a_{t-1}(\alpha)\vee\phi_{t-2}=\phi_{t-2},\,\forall\alpha\geq 0$. Indeed, $ \alpha\geq0$ implies that $a_{t-1}(\alpha)\leq a_{t-1}(0)\leq \phi_{t-2}$. Hence,\bea\label{subcase2-1}\bar\varphi_{t-1}(x)=\varphi_{t-1}(0,\phi_{t-2},x).\eea

If $a_{t-1}(0)\geq \phi_{t-2}$, we have 

\bea\label{subcase2-2} \bar\varphi_{t-1}(x)=\min(\bar\varphi_{t-1}^1(x),\bar\varphi_{t-1}^2(x))\eea with, $\bar\varphi_{t-1}^1(x):=\inf_{0\leq \alpha\leq a_{t-1}^{-1}(\phi_{t-2})}\varphi_{t-1}(\alpha,a_{t-1}(\alpha),x)$ and $\bar\varphi_{t-1}^2(x):=\inf_{\alpha\geq a_{t-1}^{-1}(\phi_{t-2})}\varphi_{t-1}(\alpha,\phi_{t-2},x)$. On one hand, $\bar\varphi_{t-1}^2(x)=\varphi_{t-1}(a_{t-1}^{-1}(\phi_{t-2}),\phi_{t-2},x)\geq\bar\varphi_{t-1}^1(x)$. , \bean \bar\varphi_{t-1}(x) =\bar\varphi_{t-1}^1(x)\eean
 %$\inf_{\alpha\geq a_{t-1}^{-1}(\phi_{t-2})}\varphi_{t-1}(\alpha,a_{t-1}(\alpha),x)\geq \bar\varphi_{t-1}^1(x)$. In fact, $\alpha\geq a_{t-1}^{-1}(\phi_{t-2})\iff a_{t-1}(\alpha)\leq \phi_{t-2}$. Hence, 
On the other hand, since $\eta_{t-1}^1(x)\le 0$,  we have  \bean \inf_{\alpha\geq a_{t-1}^{-1}(\phi_{t-2})}\varphi_{t-1}(\alpha,a_{t-1}(\alpha),x)\geq \varphi_{t-1}(a_{t-1}^{-1}(\phi_{t-2}),\phi_{t-2},x)  \geq \bar\varphi_{t-1}^1(x)\eean
Therefore,  $\bar\varphi_{t-1}^1(x)=\inf_{\alpha\geq0}\varphi_{t-1}(\alpha,a_{t-1}(\alpha),x)=\bar\varphi_{t-1}(x)$.

Let $x\geq m_{t-1}^{1-}$, then $a\to \varphi_{t-1}(\alpha,a,x)$ is non-decreasing. Thus, \bea\label{case3} \bar\varphi_{t-1}(x)=\inf_{\alpha\geq 0}\varphi_{t-1}(\alpha,a_{t-1}(\alpha),x).\eea

Finally, if $x>M_{t-1}^{N-}$, we have from above that $\bar\varphi_{t-1}(x)=\inf_{\alpha\geq 0}\varphi_{t-1}(\alpha,a_{t-1}(\alpha),x)$. Explicitly, $\varphi_{t-1}(\alpha,a_{t-1}(\alpha),x)$ is given by,
\bean 
 \begin{cases}
    \begin{array}{ll}
        \eta_{t-1}^1(x)\max_{i=2,\dots,2N}(b_{t-1}^{i}-\alpha p_{t-1}^i)+\kappa_{t-1}\phi_{t-2}x+\bar y_{t-1}^1 +\alpha, & \mbox{if }  \alpha\geq a_{t-1}^{-1}(\phi_{t-2})\\
        \eta_{t-1}^2(x)\max_{i=2,\dots,2N}(b_{t-1}^i-\alpha p_{t-1}^i)-\kappa_{t-1}\phi_{t-2}x+\bar y_{t-1}^1 +\alpha,& \mbox{if }\alpha\leq a_{t-1}^{-1}(\phi_{t-2})\notag
    \end{array}
\end{cases}
\eean
But $\eta_{t-1}^i,\: i=1,2$ are both positive in this case. Hence $\varphi_{t-1}(\alpha,a_{t-1}(\alpha),x)$ is equal to:
\bean
  \begin{cases}
    \begin{array}{ll}
      \max_{i=2,\dots,2N}(  \eta_{t-1}^1(x)b_{t-1}^i+(1- \eta_{t-1}^1(x) p_{t-1}^i)\alpha)+\kappa_{t-1}\phi_{t-2}x+\bar y_{t-1}^1,  & \mbox{if } \alpha\geq a_{t-1}^{-1}(\phi_{t-2}) \\
       \max_{i=2,\dots,2N}( \eta_{t-1}^2(x)b_{t-1}^i+(1- \eta_{t-1}^2(x)p_{t-1}^i)\alpha)-\kappa_{t-1}\phi_{t-2}x+\bar y_{t-1}^1, & \mbox{if } \alpha\leq a_{t-1}^{-1}(\phi_{t-2}) 
    \end{array}
    \end{cases}
\eean 
Since $x$ is strictly greater than $M_{t-1}^{N-}$, it is in particular also strictly greater than $M_{t-1}^{N+}$ implying that $\eta_{t-1}^i(x)> M_{t-1}^N-m_{t-1}^1$ for $i=1,2$. In conclusion, since $p_{t-1}^2\leq\dots \leq p_{t-1}^{2N}=\frac{1}{M_{t-1}^N-m_{t-1}^1}$, we get that, $\eta_{t-1}^i(x)p_{t-1}^i>1$. In other words, the slopes in $\alpha$ are strictly negative hence $\bar\varphi_{t-1}(x)=-\infty$ by taking $\alpha=+\infty$. 
\end{proof}

\begin{theo}\label{slopes order}
    The random function $\bar h_{t-1}$ defined by \eqref{barh_{t-1}} is a convex, piecewise affine function on the interval $[m_{t-1}^{1+}, M_{t-1}^{N-}]$. Moreover, on each subinterval where $\bar h_{t-1}$ is affine, the slopes are of the form $-\hat\Phi_{\phi_{t-2}}(a_{t-1}^{\alpha_i})$, for $i = 1, \dots, d$, with $d \ge 1$. Finally, these slopes are strictly smaller than $-\phi_{t-2}$ on the first interval $[m_{t-1}^{1+}, m_{t-1}^{1-}]$ and strictly larger than $-\phi_{t-2}$ on the last interval $[M_{t-1}^{N+}, M_{t-1}^{N-}]$.
\end{theo}
\begin{proof}
     Recall that $\bar h_{t-1}=-\bar \varphi_{t-1}(x)$ where $\bar \varphi_{t-1}$ is concave as an infimum of affine function in $x$. It follows that $\bar h_{t-1}$ is convex.  Let us show that $\bar \varphi_{t-1}$ is a piecewise affine function. As a first step, we show it on the interval $[m_{t-1}^{1-},M_{t-1}^{N-}]$.  Recall that, by the proof  of Theorem \ref{optimization} , we have  $\bar \varphi_{t-1}=\min_{i=1,2}\bar \varphi_{t-1}^i$ where 
\bea  \bar \varphi_{t-1}^1(x)&=&\inf_{\alpha\in [0,a_{t-1}^{-1}(\phi_{t-2})\vee 0]}\varphi_{t-1}(\alpha,a_{t-1}(\alpha),x),\\ \label{phi-1}
\bar \varphi_{t-1}^2(x)&=&\inf_{\alpha \ge a_{t-1}^{-1}(\phi_{t-2})\vee 0}\varphi_{t-1}(\alpha,a_{t-1}(\alpha),x).
\eea
It is then sufficient to show that $\bar \varphi_{t-1}^i$, $i=1,2$, are piecewise affine function. To see it, notice that by defintion of $a_{t-1}$ the functions $\bar\varphi_{t-1}^i$ are of the form 

$\bar\varphi_{t-1}^i(x)=\inf_{\alpha\in [\alpha_0^i,\alpha_1^i]}\max_{j=1,\cdots,N}(a_{j,i}^x\alpha +b_{j,i}^x)$ for some coefficients $a_{j,i}^x,b_{j,i}^x$ which are affine functions of $x$. Precisely, we have:
\bean a_{j,i}^x&=&1-p_{t-1}^j((1+(-1)^{i+1} \kappa_{t-1})x-m_{t-1}^1),\\
b_{j,i}^x&=&\bar y_{t-2}^1+(-1)^i \kappa_{t-1}\phi_{t-2}x+((1+(-1)^{i+1} \kappa_{t-1})x-m_{t-1}^1)b_{t-1}^j.
\eean

Notice that the bounds $\alpha_0^i,\alpha_1^i$, $i=1,2$, does not depend on $x$. Let us denote by $x^{j,i}\in \R$ the solutions to the equations $a_{j,i}^{x^{j,i}}=0$. We suppose w.l.o.g. that $x^{j,i}<x^{j+1,i}$ for all $i,j$.

In the case where $x\ge \max_{j=1,\cdots,N}x^{j,i}=x^{N,i}$ for $i=1$ (resp. $i=2$), we have $a_{j,i}^x\ge 0$ for all $j=1,\cdots,N$ hence 
$\bar\varphi_{t-1}^i(x)=\max_{j=1,\cdots,N}(a_{j,i}^x\alpha_1^i +b_{j,i}^x)$ so that $\bar\varphi_{t-1}^i$, $i=1$ (resp. $i=2$), is a piecewise affine function. 

In the case where $x\le \min_{j=1,\cdots,N}x^{j,i}=x^{1,i}$ for $i=1$ (resp. $i=2$), we have $a_{j,i}^x\le 0$ for all $j=1,\cdots,N$ hence 
$\bar\varphi_{t-1}^i(x)=\max_{j=1,\cdots,N}(a_{j,i}^x\alpha_0^i +b_{j,i}^x)$ so that $\bar\varphi_{t-1}^i$, $i=1,2$ are piecewise affine functions. 

Otherwise, for each $i=1,2$, if $x\in (x^{1,i},x^{N,i})$,  some slopes $a_{j,i}^x$ are non-positive and at least one is strictly positive. Therefore, by Lemma \ref{KeyTool}, each $\bar\varphi_{t-1}^i$, $i=1,2$, coincides with a maximum of  functions which are affine in $x$. We just need to verify that the set of all non-negative (resp. positive) slopes $a_{k,i}^x$ corresponds to a fixed set of indices, i.e. a set that does not depend on $x$, when   $x\in [x^{j,i},x^{j+1,i}]$ for some given $j\le N-1$. This is clear since we  have $a_{k,i}^x\le 0$ if and only if $k\le j$ and $a_{m,i}^x>0 $ if and only if $m\ge j+1$. By Lemma \ref{KeyTool}, we then conclude that each $\bar\varphi_{t-1}^i$ is a piecewise affine function
 on every interval $[x^{j,i},x^{j+1,i}]$, $j=1,\cdots,N-1$, as a maximum of affine functions in $x$ over a set of indices that only depend on $j$ when $x\in [x^{j,i},x^{j+1,i}]$.

On the interval $[m_{t-1}^{1+},m_{t-1}^{1-}]$, the reasoning is similar, i.e. by the proof of Theorem \ref{optimization} , we also have $\bar \varphi_{t-1}=\min_{i=1,2}\bar \varphi_{t-1}^i$ if $a_{t-1}^{-1}(\phi_{t-2})\ge 0$. Otherwise,  $\bar \varphi_{t-1}(x)=\varphi_{t-1}(0,\phi_{t-2},x)$, which proves that $\bar h_{t-1}$ is a piecewise affine function. 

The second  statement is a direct consequence of  the expression of $\varphi_{t-1}$. Indeed, the proof of  Theorem \ref{optimization} shows that, for each $x$, the infimum $\bar \varphi_{t-1}(x)$ (over all $\alpha\ge 0$) is attained by some $\alpha^x$ which is necessarily a constant (i.e. independent of $x$) on each interval on which $\bar \varphi_{t-1}$ is an affine function.

At last,  on the interval $[m_{t-1}^{1+},m_{t-1}^{1-}]$,  $\bar h_{t-1}=-\varphi_{t-1}(\alpha,a_{t-1}(\alpha)\vee\phi_{t-2},x)$, i.e. either a slope is  $-\hat\Phi_{\phi_{t-2}}(\phi_{t-2})=-\phi_{t-2}$, $\phi_{t-2}=a_{t-1}(\alpha_0)$ with $\alpha_0=a_{t-1}^{-1}(\phi_{t-2})$  or a slope is $-\hat\Phi_{\phi_{t-2}}(a_{t-1}^{\alpha_i})$ with $a_{t-1}^{\alpha_i}\ge \phi_{t-2}$. Therefore, $\hat\Phi_{\phi_{t-2}}(a_{t-1}^{\alpha_i})\ge \hat\Phi_{\phi_{t-2}}(\phi_{t-2})$ so that the slope is smaller than $-\hat\Phi_{\phi_{t-2}}(\phi_{t-2})=-\phi_{t-2}$. \smallskip

On the interval $[M_{t-1}^{N+},M_{t-1}^{N-}]$, recall that $\bar \varphi_{t-1}=\min_{i=1,2}\bar \varphi_{t-1}^i$ and we have $\bar \varphi_{t-1}^1=\inf_{\alpha\in [0,a_{t-1}^{-1}(\phi_{t-2})\vee 0]}\varphi_{t-1}(\alpha,a_{t-1}(\alpha),x)$. If $a_{t-1}(0)\le \phi_{t-2}$, then we have $a_{t-2}^{-1}(\phi_{t-2})\le 0$ hence we have $\bar \varphi_{t-1}^1(x)=\varphi_{t-1}(0,a_{t-1}(0),x)$. We also deduce that $\bar \varphi_{t-1}^2(x)=\inf_{\alpha \ge  0]}\varphi_{t-1}(\alpha,a_{t-1}(\alpha),x)\le \bar \varphi_{t-1}^1(x)$. Therefore, $\bar \varphi_{t-1}(x)=\bar \varphi_{t-1}^2(x)$. This implies, by definition of $\bar \varphi_{t-1}^2$ that the slopes of $\bar h_{t-1}$ are of the form $-\hat\Phi_{\phi_{t-2}}(a_{t-1}^{\alpha_i})$ with $\alpha_i\ge 0$ such that $a_{t-1}^{\alpha_i}\le a_{t-1}^{0}\le \phi_{t-2}$. So,  $-\hat\Phi_{\phi_{t-2}}(a_{t-1}^{\alpha_i})\ge -\phi_{t-2}$. In the case where $a_{t-1}(0)> \phi_{t-2}$, we have $\bar \varphi_{t-1}^1(x)=\inf_{\alpha\in [0,a_{t-1}^{-1}(\phi_{t-2})]}\varphi_{t-1}(\alpha,a_{t-1}(\alpha),x)$. Note that   $\varphi_{t-1}(\alpha,a_{t-1}(\alpha),x)$ is piecewise affine function in $\alpha$ and the slopes are of the form $\epsilon^i=1-\bar p_i\left(x(1+\kappa_{t-1})-m^1_{t-2}  \right)$ where $(\bar p_i)_i$ are defined in (\ref{pente-pi}). We get that $\epsilon^i\ge 0$ if and only if $x\ge (\bar p_i)^{-1}$ where $(\bar p_i)^{-1}\in \{m_{t-1}^{i+},M_{t-1}^{i+}\}$. Since $M_{t-1}^{N+}=\max_{i=1,\cdots,N}\{m_{t-1}^{i+},M_{t-1}^{i+}\}$, we deduce that all the slopes $\epsilon^i$  are non-positive hence $\bar \varphi_{t-1}^1(x)=\varphi_{t-1}(a_{t-1}^{-1}(\phi_{t-2}),\phi_{t-2},x)$. This implies that $\bar \varphi_{t-1}^2(x)\le \bar \varphi_{t-1}^1(x)$ hence $\bar \varphi_{t-1}(x)=\bar \varphi_{t-1}^2(x)$ and we may conclude as in the previous case.

\end{proof}

\begin{rem} \label{partition}

By Proposition \ref{optimization}, there exists $d \in \mathbb{N}$ and a partition
\[
[m_{t-1}^{1+}, M_{t-1}^{N-}] = \bigcup_{i=0}^{d-1} [O_i, O_{i+1}]
\]
with endpoints $O_0 = m_{t-1}^{1+}$ and $O_d = M_{t-1}^{N-}$, such that the slope of the piecewise affine function $\bar h_{t-1}$ is $\hat\Phi_{\phi_{t-3}}(a_{t-1}^{\alpha_i})$ on each interval $[O_i, O_{i+1}]$. We deduce that the  Fenchel conjugate, $\bar h_{t-1}^*$, is also piecewise affine. Its slopes correspond to the points $O_0, \dots, O_d$ over the partition of $\mathbb{R}$ induced by the increasing sequence $\bigl(-\hat\Phi_{\phi_{t-2}}(a_{t-1}^{\alpha_i})\bigr)_{i=1}^d$.
\end{rem}
To compute $h_{t-1}$ given by \eqref{h_{t-1}}, we shall see that $\bar{h}_{t-1}^* \circ \Phi_{\phi_{t-2}}$ is convex and, therefore, we have: 
\[
h_{t-1} = \left[ (\bar{h}_{t-1}^* \circ \Phi_{\phi_{t-2}})^{**} \circ \Phi^{-1}_{\phi_{t-2}} \right]^* = \bar{h}_{t-1}.
\]

\begin{lemm}\label{teta} With the notations of Remark \ref{partition}, suppose that there exists $j$ such that $-\hat\Phi_{\phi_{t-2}}(a_{t-1}^{\alpha_j})\leq  -\phi_{t-2}$ and $-\hat\Phi_{\phi_{t-2}}(a_{t-1}^{\alpha_{j+1}})>  -\phi_{t-2}$. Then, there exits $i$ such that  $\phi_{t-2}=\hat\Phi_{\phi_{t-2}}(a_{t-1}^{\alpha_i})$.
\end{lemm}
\begin{proof}
    The assumption is equivalent to $a_{t-1}^{\alpha_j}\geq \phi_{t-2}$ and  $a_{t-1}^{\alpha_{j+1}}< \phi_{t-2}$. We deduce 
    $\lambda \in]0,1]$ such that, $ \phi_{t-2}=\lambda a_{t-1}^{\alpha_j} + (1-\lambda) a_{t-1}^{\alpha_{j+1}}.$ By left and right continuity of $\bar\varphi_{t-1}$ at point $O_{j+1}$, we get that:
\bea\label{0}
\bar\varphi_{t-1}(O_{j+1})&=&\lambda \bar\varphi_{t-1}(O_{j+1}) + (1-\lambda)\bar\varphi_{t-1}(O_{j+1})\quad \\ \notag
                          &=&  \lambda \bar\varphi_{t-1}(O_{j+1}^-)+ (1-\lambda)\bar\varphi_{t-1}(O_{j+1}^+)\\ \notag &=&\lambda \bar\varphi_{t-1}^2(O_{j+1})+ (1-\lambda)\bar\varphi_{t-1}^1(O_{j+1}).
\eea

Let us introduce the notation $C_{\lambda}(a_j)=\lambda a_j+(1-\lambda)a_{j+1}$ for the convex combination of  any pair of real numbers $(a_j,a_{j+1})$ and $j\in \N$. Using the explicit affine expression of $\varphi_{t-1}(\alpha,a,x)$ and the definitions of $\alpha_j,\alpha_{j+1}$, we get that 
\bea \label{IFunda} \notag \bar\varphi_{t-1}(O_{j+1})&=& \lambda \varphi_{t-1}(\alpha_j,a_{t-1}^{\alpha_j},O_{j+1}) + (1-\lambda)\varphi_{t- 1}(\alpha_{j+1},a_{t-1}^{\alpha_{j+1}},O_{j+1})\\
&=&C_{\lambda}(\hat\Phi_{\phi_{t-2}}(a_{t-1}^{\alpha_j}))O_{j+1}+y^1_{t-1}-\phi_{t-2}m^1_{t-1}+C_{\lambda}(\alpha_{j}).
\eea

Since $\hat\Phi_{\phi_{t-2}}$ is convex, we have:
\bea\label{1}
C_{\lambda}(\hat\Phi_{\phi_{t-2}}(a_{t-1}^{\alpha_j}))&\geq& \hat\Phi_{\phi_{t-2}}(\lambda a_{t-1}^{\alpha_j} + (1-\lambda) a_{t-1}^{\alpha_{j+1}})=\hat\Phi_{\phi_{t-2}}(\phi_{t-2}).
\eea
Moreover, by convexity of the mapping $\alpha\mapsto a_{t-1}(\alpha)$, we have: \[a_{t-1}(\lambda \alpha_j+(1-\lambda)\alpha_{j+1})\leq \lambda a_{t-1}^{\alpha_j} + (1-\lambda) a_{t-1}^{\alpha_{j+1}}=\phi_{t-2}\]
As $a_{t-1}$ is decreasing, we obtain
\bea\label{2} C_{\lambda}(\alpha_{j})\geq a_{t-1}^{-1}(\phi_{t-2}).\eea
Using the inequalities $(\ref{1})$ and $(\ref{2})$, we conclude by (\ref{IFunda}) that 
\[\bar\varphi_{t-1}(O_{j+1})\geq \hat\Phi_{\phi_{t-2}}(\phi_{t-2})O_{j+1}+y^1_{t-1}-\phi_{t-2}m^1_{t-1}+a_{t-1}^{-1}(\phi_{t-2})\]
i.e.\[\bar\varphi_{t-1}(O_{j+1})\geq \bar\varphi_{t-1}^1(O_{j+1})=\bar\varphi_{t-1}^2(O_{j+1})=\varphi_{t-1}(a_{t-1}^{-1}(\phi_{t-2}),\phi_{t-2},O_{j+1}).\]

On the other hand, by Proposition \ref{optimization}, we have $O_{j+1}\ge m_{t-1}^{1-}$.  Therefore,  we deduce that $\bar\varphi_{t-1}(O_{j+1})=\min(\min_{\alpha\in\tau^1}\bar\varphi_{t-1}^1,\min_{\alpha\in\tau^2}\bar\varphi_{t-1}^2),$ still by  Proposition \ref{optimization}, where $a^{-1}_{t-1}(\phi_{t-2})\in \tau^1\cap \tau^2$.  So, necessarily, we have  $$\bar\varphi_{t-1}(O_{j+1})=\bar\varphi_{t-1}^1(O_{j+1})=\bar\varphi_{t-1}^2(O_{j+1})$$ and $ \alpha_j=a_{t-1}^{-1}(\phi_{t-2})$, i.e.  $\hat\Phi_{\phi_{t-2}}(a_{t-1}^{\alpha_j})=\phi_{t-2}$ as stated.
\end{proof}

 \begin{theo}\label{ConvexityFundamental}The function $\bar{h}_{t-1}^* \circ \Phi_{\phi_{t-2}}$ is convex on $\R$.\end{theo}
\begin{proof}
    Recall that, by convexity, the slopes of $\bar{h}_{t-1}^*$ are non decreasing and coincide with the elements $O_i$ of the partition defining $\bar{h}_{t-1}$, see (\ref{partition}). Recall that by the proof of Theorem \ref{slopes order}, the slopes of $\bar{h}_{t-1}$ are $x_i=-\hat \Phi_{\phi_{t-2}}(a_{t-1}^{\alpha_i})$. They  are also non decreasing and define a partition for $\bar{h}_{t-1}^*$. The smallest one is smaller than $-\phi_{t-2}$ with the smallest $\alpha_i\le a_{t-1}^{-1}(\phi_{t-2})$ while the largest one is larger than $-\phi_{t-2}$ with the largest  $\alpha_i\ge a_{t-1}^{-1}(\phi_{t-2})$. \smallskip
 
 We first suppose that  $a_{t-1}^{0}\geq \phi_{t-2}$ hence $-\hat \Phi_{\phi_{t-2}}(a_{t-1}^{0})\leq -\phi_{t-2}$. Note  that $\bar\varphi_{t-1}$ coincides with $\bar\varphi_{t-1}^2$, see the proof of Theorem \ref{slopes order}, for  $x<M^{N-}_{t-1}$ sufficiently closed to $M^{N-}_{t-1}$,  which implies that the last slope is larger than $-\phi_{t-2}$.     Let  us consider the interval $[O_j,O_{j+1}]$ on which   $\bar{h}_{t-1}$ admits the  largest slope $x_j$ smaller than $-\phi_{t-2}$ so that the  slope of $\bar{h}_{t-1}$  is  $x_{j+1}>-\phi_{t-2}$ on the next interval $[O_{j+1},O_{j+2}]$.
This is possible if and only if at least one slope of $\bar{h}_{t-1}$  is strictly larger than $-\phi_{t-2}$. The case where all the slopes are smaller (resp. larger) than $-\phi_{t-2}$  will be considered later.\smallskip

We deduce that $\hat{O}_j^+ = (1 + \kappa_{t-1}) O_j$ and $\hat{O}_{j+1}^- = (1 - \kappa_{t-1}) O_{j+1}$ are the corresponding slopes for $\bar{h}_{t-1}^* \circ \Phi_{\phi_{t-2}}$ on the interval $[x_j,-\phi_{t-2}]$ and $[-\phi_{t-2},x_{j+1}]$ respectively. It suffices to check that $\hat{O}_j^+ \leq \hat{O}_{j+1}^-$ to deduce that $\bar{h}_{t-1}^* \circ \Phi_{\phi_{t-2}}$ is convex on $\R$. Indeed, the function $\bar{h}_{t-1}^* \circ \Phi_{\phi_{t-2}}$ is convex on $]-\infty, -\phi_{t-2}[$ and $]-\phi_{t-2}, +\infty[$, respectively. To see it, note that, $\bar{h}_{t-1}^*$ is convex on $\mathbb{R}$ and $\Phi_{\phi_{t-2}}$ is affine on each interval $]-\infty, -\phi_{t-2}[$ and $]-\phi_{t-2}, +\infty[$, respectively. \smallskip

By Lemma \ref{teta}, we may assume that $x_j=-\phi_{t-2}$, i.e.  $x_j=-\hat\Phi_{\phi_{t-2}}(a_{t-1}^{\alpha_j})$, where $\alpha_j = a_{t-1}^{-1}(\phi_{t-2})$ and the previous slope $x_{j-1}=-\hat\Phi_{\phi_{t-2}}(a_{t-1}^{\alpha_{j-1}}) <- \phi_{t-2}$ with  $\alpha_{j-1}<a_{t-1}^{-1}(\phi_{t-2})$. Note that $\bar{h}_{t-1}=-\bar\varphi_{t-1}$ coincides with $\bar{h}_{t-1}^1:=-\bar\varphi_{t-1}^1$ on $(-\infty,O_{j+1}]$ and with $\bar{h}_{t-1}^2:=-\bar\varphi_{t-1}^2$ on $[O_{j},\infty)$.

By  left continuity at point $O_j$, $\bar\varphi_{t-1}^1(O_j)$ and its left limit $\bar\varphi_{t-1}^1(O_j-)$ coincide, i.e.:
\bea\label{left}\notag
\bar\varphi_{t-1}(O_j)&=&\bar\varphi_{t-1}^1(O_j)=\bar\varphi_{t-1}^1(O_j-)=\varphi_{t-1}(\alpha_{j-1},a_{t-1}^{\alpha_{j-1}},O_j-),\\\notag
                    &=&\hat\Phi_{\phi_{t-2}}(a_{t-1}^{\alpha_{j-1}})O_j+y^1_{t-1}-a_{t-1}^{\alpha_{j-1}}m^1_{t-1}+\alpha_{j-1},\\
                    &=&[(1+\kappa_{t-1})a_{t-1}^{\alpha_{j-1}}-\kappa_{t-1}\phi_{t-2}]O_j+y^1_{t-1}-a_{t-1}^{\alpha_{j-1}}m^1_{t-1}+\alpha_{j-1}.\quad\quad\quad
\eea
Since $\bar\varphi_{t-1}$ is affine on $[O_j,O_{j+1}]$ with slope $\hat\Phi_{\phi_{t-2}}(a_{t-1}^{\alpha_j})=-\phi_{t-2}$, we have  
\bean \bar\varphi_{t-1}(O_{j+1})&=&\bar\varphi_{t-1}(O_j)+\hat\Phi_{\phi_{t-2}}(a_{t-1}^{\alpha_j})(O_{j+1}-O_j).
\eean
Therefore, by left continuity, $\bar\varphi_{t-1}(O_j)+\hat\Phi_{\phi_{t-2}}(a_{t-1}^{\alpha_j})(O_{j+1}-O_j)=\bar\varphi_{t-1}^1(\alpha_j,O_{j+1}-)$. Using  (\ref{left}) and the affine expression of $\bar\varphi_{t-1}$ with slope $\phi_{t-2}$ on $[O_j,O_{j+1}]$, the previous equality implies that   \[\hat O_j^+=m_{t-1}^1+\frac{\alpha_j-\alpha_{j-1}}{a_{t-1}^{\alpha_{j-1}}-a_{t-1}^{\alpha_j}}.\]

Similarly, by continuity at point $O_{j+1}$, we may prove that \[\hat O_{j+1}^-=m_{t-1}^1+\frac{\alpha_j-\alpha_{j+1}}{a_{t-1}^{\alpha_{j+1}}-a_{t-1}^{\alpha_j}}.\]
Thus, by  convexity of the map $\alpha \mapsto a_{t-1}(\alpha) = a_{t-1}^{\alpha}$, we conclude that  we have $\hat{O}_j^+ \leq \hat{O}_{j+1}^-$. \bigskip

Let us now consider the case where $x_i\ge -\phi_{t-2}$ for all $i$. This implies that $x_i=-\hat \Phi_{\phi_{t-2}}(a_{t-1}^{\alpha_i})$ with $a_{t-1}^{\alpha_i}\le \phi_{t-2}$ or $\alpha_i\ge a_{t-1}^{-1}(\phi_{t-2})$. Therefore, the first (and smallest) slope is necessary $-\phi_{t-2}$ and  $\bar\varphi_{t-1}=\bar\varphi_{t-1}^2$. Indeed, see the expression of $\bar\varphi_{t-1}^2$ in (\ref{phi-1}) which proves that the slope $-\phi_{t-2}$ is attained. Recall that $\bar\varphi_{t-1}=\bar\varphi_{t-1}^1$ on $[m_{t-1}^{1+}, m_{t-1}^{1-}]$, by the proof of Theorem \ref{slopes order} , hence the first interval $[O_j,O_{j+1}]$ on which the slope of $\bar{h}_{t-1}$ is  $-\phi_{t-2}$  is such that $O_j=m_{t-1}^{1+}$ hence $\hat O_j^+=m_{t-1}^{1}$. Moreover, we necessarily have $O_{j+1}\ge m_{t-1}^{1-}$, which implies that $\hat O_{j+1}^-\ge m_{t-2}^{1}$ hence $\hat O_j^-\ge \hat O_j^+$ as desired. {\bf Note that this case also allows to conclude when $a_{t-1}^0<\phi_{t-2}$}. Indeed, we then have $a_{t-1}^\alpha<\phi_{t-2}$ for any $\alpha\ge 0$ hence $x_i\ge -\phi_{t-2}$ for all $i$.\smallskip

At last, it remains to consider the case where $x_i\leq -\phi_{t-2}$ for all $i$. This implies that $x_i=-\hat \Phi_{\phi_{t-2}}(a_{t-1}^{\alpha_i})$ with $a_{t-1}^{\alpha_i}\ge \phi_{t-2}$ or $\alpha_i\le a_{t-1}^{-1}(\phi_{t-2})$. Therefore, the last (and largest) slope is necessary $-\phi_{t-2}$, see Theorem \ref{slopes order}, and  $\bar\varphi_{t-1}=\bar\varphi_{t-1}^1=\bar\varphi_{t-1}^2$ on the interval $[O_j,O_{j+1}]$. Moreover, on this interval,  the slope of $\bar{h}_{t-1}$ is  $-\phi_{t-2}$  so that that $O_{j+1}=M_{t-1}^{N-}$. This implies that $\hat O_{j+1}^-\ge M_{t-1}^{N}$. On the other hand, by Theorem \ref{slopes order}, we have $O_j\le M_{t-1}^{N+}$ hence $O_j^+\le M_{t-1}^{N}\le O_{j+1}^-$. The conclusion follows. 
\end{proof}

\begin{coro} \label{coro-convex-infprice} 
We have:
\bean
\left((\bar h_{t-1}^*\circ \Phi_{\phi_{t-2}})^{**}\circ \phi_{\phi_{t-2}}^{-1}\right)^*&=& \bar h_{t-1}.
\eean
\end{coro}

\begin{proof}Since $\bar h_{t-1}^*\circ \Phi_{\phi_{t-2}}$ is convex, it follows that $(\bar h_{t-1}^*\circ \Phi_{\phi_{t-2}})^{**}=\bar h_{t-1}^*\circ \Phi_{\phi_{t-2}}$ and the conclusion follows from the convexity of $\bar h_{t-1}$. \end{proof}

To avoid that the minimal super-hedging price is $p_{t-1}(g_t,\phi_{t-2}) = -\infty$, a relative AIP condition at time $t-1$ with respect to the choice of $g_t$ is required. Notably, this condition does not depend on $\phi_{t-2}$, as stated in Definition \ref{AIPrel}.

\subsection{Conclusion: Infimum price under AIP and associated optimal strategy}\label{GeneralStepProc}
\bigskip

From the previous subsection, we deduce the main result, i.e. Theorem \ref{theo:payoff}, which shows that the infimum super-hedging price at time $t-1$ remains a payoff function of $S_{t-1}$, in the same structural form as $g_t$. In fact, we prove that this infimum coincides with the minimal super-hedging price.

Propositions \ref{Prop-Conv11} and \ref{Prop-Conv21} give explicit expressions for the payoff function $g_{t-1}$ in terms of $g_t$, as derived from the proof of Theorem \ref{theo:payoff}.

Finally, Propositions \ref{OptStrat11} and \ref{OptStrat21} provide the optimal trading strategies corresponding to these minimal super-hedging prices.

% ===== Section 5 Appendix =====
\section{Appendix}\label{A}
%%%%%%%%%%%%%%% subsection: A %%%%%%%%%%%%%%%%%%%%%%%%%%%%%%
\subsection{Conditional essential infimum and supremum}
The definitions of conditional essential supremum and infimum of a family of random variables are given in \cite[Proposition 2.5]{CL}.

\begin{theo} Let $\cH$ and $\cF$ be complete $\sigma$-algebras such that $\cH\subseteq \cF$ and let  $\Gamma=(\gamma_i)_{i\in I}$ be a family of real-valued $\cF$-measurable random variables.  There exists a unique $\R\cup\{+\infty\}$-valued $\cH$-measurable random variable, denoted by $\esssup_{\cH} \Gamma$, such that $\esssup_{\cH} \Gamma\ge \gamma_i$ a.s. for all $i\in I$ and,   if $\bar \gamma$ is $\cH$-measurable and satisfies $\bar \gamma \ge \gamma_i$ a.s. for all $i\in I$, then  $\bar \gamma \ge \esssup_{\cH} \Gamma$ a.s..
\end{theo}
Recall that the conditional support $\supp_{\cH} X$ of a random variable $X$ is defined as the smallest $\cH$-measurable random set that contains $X$ a.s., see \cite{ELM}. The following proposition is a key tool for our approach, see proof in given in \cite{CL}.
\begin{prop}\label{esssup-sup}
  Let $X\in L^0(\R,\cF)$  and let $h: \Omega\times \R\mapsto \R$ be a $\cH\otimes\cB(\R)$-measurable function which is lower semi-continuous ($l.s.c.$) in x. Then,
  \bea
  \esssup_{\cH}h(X) = \sup_{x\in \supp_{\cH}X}h(x)\; a.s..
  \eea

\end{prop}

Recall that, if $h$  is a $\cH$-normal integrand on $\R$ (see Definition 14.27 in \cite{rw}) then $h$ is $\cH \otimes \cB(\R)$-measurable and is l.s.c. in $x$, see  \cite[Definition 1.5]{rw}),  and the converse holds true if $\cH$ is complete for some probability measure, see \cite[Corollary 14.34]{rw}. Similarly, we have the following result, see Lemma 5.5 \cite{Lep2}.

\begin{prop}\label{PropEssinfNI}
For any $\cH$-normal integrand $f$ on $\R$, we have
$$\essinf_{\cH}\{f(A):~A\in L^0(\R,\cH)\}=\inf_{a\in \R}f(a).$$
\end{prop}

\subsection{Auxiliary results}

\begin{lemm}\label{LemmConvexityComp}
Let $f$ be a function from $\R$ to $[-\infty,\infty]$ such that $f=\infty$ on $\R_{-}$. For every   real-valued convex function $\varphi$,  $f^*\circ \varphi$ is a convex function.  
\end{lemm}
\begin{proof}
Since $\varphi$ is a convex function,  the mapping $x\mapsto \varphi(x)y-f(y)$ is convex, for every fixed $y\in \R_+$. Observe that

$$f^*\circ \varphi(x)=\sup\{\varphi(x)y-f(y),y\in \R \}=\sup\{\varphi(x)y-f(y),y\in \R_+ \},$$ so that $f^*\circ \varphi$ is convex as a pointwise supremum of convex functions.
\end{proof}

\begin{lemm}\label{LemmDualComp}
Let $\gamma\in\cA$, $\gamma(x):=ax+b$, and $\varphi$ be a bijection. Then, $(\gamma^*\circ\varphi)^*$ is an affine function given by $(\gamma^* \circ \varphi)^*(x):=\varphi^{-1}(a)x+b$ and  $(\gamma^*\circ\varphi)^{**}=\gamma^*\circ\varphi.$
\end{lemm} 
\begin{proof}
Recall that

$$\gamma^* \circ \varphi(y)=\sup\{z\varphi(y)-\gamma(z),z\in \R \}=\sup\{(\varphi(y)-a)z-b,z\in \R \}.$$
We deduce that 
$$\gamma^* \circ\varphi(y)=-b 1_{\varphi^{-1}(a)}(y)+ \infty 1_{\R\setminus \varphi^{-1}(a)}(y) $$
Therefore,
$$(\gamma^* \circ\varphi)^*(x):=\sup\{xy-\gamma^* \circ \varphi(y),\,y\in \R \}=x \varphi^{-1}(a)+b.$$
Consequently,
\begin{align*}
	(\gamma^* \circ\varphi)^{**}(y)&:=\sup\{xy-(\gamma^* \circ \varphi)^*(x),\,x\in \R \}\\
	& =\sup\{(y-\varphi^{-1}(a))x-b,x\in \R \}\\
	&= -b 1_{\varphi^{-1}(a)}(y)+ \infty 1_{\R\setminus \varphi^{-1}(a)}(y)\\
	&=\gamma^* \circ\varphi(y)
\end{align*}

\end{proof}

\begin{prop}\label{propconjf} Suppose that  $f$ is a function defined on $\R$ with values in $]-\infty,\infty]$ such that $f=\infty$ on $\R_-$. Then,  for every  bijection  $\Phi$ such that $\Phi^{-1}$ is real-valued and convex,  there exists a unique lower semi-continuous convex function $h$ such $h^*\circ \Phi$ is  l.s.c.,  convex and satisfies $f^{**}=(h^*\circ \Phi)^*$. 
Moreover, $h=(f^*\circ \Phi^{-1})^*$ and we have: 
\bean
h=[(h_1^*\circ \Phi)^{**}\circ \Phi^{-1}]^*,
\text{ where }  h_1(x):=\sup\{(\gamma^*\circ \Phi^{-1})^*(x), \gamma\in\cA\,\, {\rm and\,} \,\gamma\le f\}.
\eean
\end{prop}

\begin{proof}
Since $f=\infty$ on $\R_{-}$, by Lemma \ref{LemmConvexityComp}, $f^*\circ \Phi^{-1}$ is a convex function. As it is also l.s.c,  $(f^*\circ \Phi^{-1})^{**}=f^*\circ \Phi^{-1}$. Hence, with  $h=(f^*\circ \Phi^{-1})^*$, we  have $(h^*\circ \Phi)^*=f^{**}$. Moreover, $f^*\circ \Phi^{-1}$ is lower semi-continuous and convex by Lemma \ref{LemmConvexityComp}. Therefore, $h^*=f^*\circ \Phi^{-1}$ and $h^*\circ \Phi=f^*$ is  lower semi-continuous and convex.

 Uniqueness follows from the second statement. To see it,  let us consider an arbitrary $\gamma \in\cA$ such that $\gamma \le f$ . Then, we deduce that $\gamma=\gamma^{**}\le f^{**}$ hence  $(h^*\circ \Phi)^*\ge \gamma$. Since $h^*\circ \Phi$ is  lower semi-continuous and convex by assumption, we deduce that $(h^*\circ \Phi)^{**}=h^*\circ \Phi$. This implies that $h\ge (\gamma^*\circ \Phi^{-1})^*$ . Taking the supremum on  every  $\gamma \in\cA$ such that $\gamma \le f$, we deduce that $h\ge h_1$. Considering the biconjugate in both sides, we deduce that $(h^*\circ \Phi)^*\ge (h_1^*\circ \Phi)^*$.

On the other hand, for every $\gamma \in\cA$ such that $\gamma \le f$, note  that by definition $h_1\ge (\gamma^*\circ \Phi^{-1})^*$. Hence $h_1^* \le \gamma^*\circ \Phi^{-1}$ by Lemma \ref{LemmDualComp}. Then $(h_1^*\circ \Phi)^*\ge \gamma$. 
Taking the supremum over all $\gamma$, we have  $(h_1^*\circ \Phi)^*\ge f^{**}$. This is equivalent to $(h_1^*\circ \Phi)^*\ge (h^*\circ \Phi)^*$.
Finally, from the first part we get $(h_1^*\circ \Phi)^*=(h^*\circ \Phi)^*$, which holds when $h=[(h_1^* \circ \Phi )^{**}\circ \Phi^{-1})^*]$.
\end{proof}

\begin{lemm}\label{KeyTool}
Let the function $T(\alpha)=\max_{j=1,\cdots,P}T^j(\alpha)$ be such that  the functions $T^j(\alpha)=a_j\alpha+b_j$ are distinct affine functions with distinct slopes. Suppose that there exist $i,j$ such that $a_i\le 0$ and $a_j>0$. Then, for any $\alpha_0\ge 0$ and $\alpha_1\ge \alpha_0$, we have:
\bean 
\inf_{\alpha \in \R}T(\alpha)&=&\max_{a_i\le 0,a_j>0}T_i(I_{i,j}),\\
\inf_{\alpha \ge \alpha_0}T(\alpha)&=&\max_{a_i\le 0,a_j>0}T_i(I_{i,j})\vee \max_{a_j>0}T^j(\alpha_0),\\
\inf_{\alpha \in[\alpha_0,\alpha_1] }T(\alpha)&=&\max_{a_i\le 0,a_j>0}T_i(I_{i,j})\vee \max_{a_j>0}T^j(\alpha_0)\vee \max_{a_i\le 0}T^i(\alpha_1)
\eean
where $I_{ij}$, $i,j=1,\cdots,P$ are the solutions to the equations $T^i(I_{ij})=T^j(I_{ij})$.\smallskip

Moreover, these formulas are still valid if $a_i>0$ for every $i$ if we adopt the convention that $max(\emptyset)=-\infty$. At last, suppose that  $\max_{j=1,\cdots,P}a_j=0$ and $\max_{j=1,\cdots,P-1}a_j<0$. Then, $\inf_{\alpha \ge \alpha_0}T(\alpha)=T_{P-1}(I_{P-1,P})=T(+\infty)$.

\end{lemm}
\begin{proof} By  assumption,  $\alpha^*=\max({\rm argmin} T)$ exists and   $T^*=T(\alpha^*)=\inf_{\alpha \in \R}T(\alpha)$. Note that $T$ is stricly increasing on $[\alpha^*,\infty)$. Moreover, $\alpha^*=I_{i^*,j^*}$ for some $i^*,j^*$ such that $a_{i^*}\le 0$ and $a_{j^*}>0$. As we also have $T^*=T^i(I_{i^*,j^*})=T^i(I_{i^*,j^*})$,   $T^*\le \max_{a_i\le 0,a_j>0}T_i(I_{i,j})$. Let us now prove the reverse inequality.   To so so, consider any $I_{ij}$ such that $a_i\le 0$ and $a_j>0$. If $I_{ij}\le I_{i^*,j^*}$, we have $T^*=T(\alpha^*)\ge T^j(I_{i^*,j^*})$ by definition of $T$. As $a_j>0$, $T^j$ is non decreasing hence $T^j(I_{i^*,j^*})\ge T^j(I_{ij})$. Otherwise, $I_{ij}\ge I_{i^*,j^*}$ and, similarly,  $T^*\ge T^i(I_{i^*,j^*})\ge T^i(I_{ij})$
since $a_i\le 0$. Therefore, $T^*\ge T^i(I_{ij})$ in any case so that we finally conclude that $T^*\ge \max_{a_i\le 0,a_j>0}T_i(I_{i,j})$. \smallskip

Let us now consider $T_+^*=T(\alpha^*_+)=\inf_{\alpha \ge \alpha_0}T(\alpha)$. In the case where $\alpha^*\ge \alpha_0$, we have $\alpha^*_+=\alpha^*$ and $T_+^*=T^*$. Moreover, $T^*_+=T(\alpha^*_+)\ge T^{j}(\alpha^*_+)$ by definition of $T$ so that $T^*_+\ge T^{j}(\alpha_0)$  if $a_j>0$, i.e. $T^*_+=T^*\vee \max_{a_j>0}T^j(\alpha_0)$. 
Otherwise, if  $\alpha^*< \alpha_0$, then  $\alpha^*_+=\alpha_0$. Then, we have $T^*_+=T(\alpha_0)\ge T^j(\alpha_0)$ for any $j$, by definition of $T$. In particular, $T^*_+\ge \max_{a_j>0}T^j(\alpha_0)$. Let us show the reverse inequality.  To do so, note that by definition of $\alpha^*$, we necessarily have $T^*_+=T(\alpha_0)=T^k(\alpha_0)$ for some $k$ such that $a_k>0$ because $T$ is strictly increasing on $[\alpha^*,\infty)$.  Therefore, we have $T^*_+\le \max_{a_j>0}T^j(\alpha_0)$ and finally the equality holds.   At last, as $T$ is non decreasing on 
$[\alpha^*,\infty)\ni \alpha_0$, we have $T^*_+=T(\alpha_0)\ge T(\alpha^*)=T^*$ hence $T^*_+=T^*_+\vee T^*$. The conclusion follows.\smallskip

At last, consider $T_{++}^*=T(\alpha^*_{++})=\inf_{\alpha \in [\alpha_0,\alpha_1]}T(\alpha)$. In the case where $\alpha^*_+\le \alpha_1$, we have $T_{++}^*=T_{+}^*=T(\alpha^*_+)$ and, by definition, $T_{++}^*\ge T^i(\alpha^*_+)$ for any $i$. Since $a_i\le 0$ implies that $T^i(\alpha^*_+)\ge T^i(\alpha_1)$ we deduce that $T_{++}^*\ge \max_{a_i\le 0} T^i(\alpha_1)$. It follows that $T_{++}^*=T_{+}^*\vee \max_{a_i\le 0} T^i(\alpha_1)$. Consider the last case $\alpha^*_+>\alpha_1$. Recall that $\alpha^*_+=\max(\alpha^*,0)$ hence we necessarily have $\alpha^*_+=\alpha^*>0$. Since $T$ is non increasing on $(-\infty, \alpha^*]$, we deduce that $T_{++}^*=T(\alpha_1)\ge T(\alpha^*)=T^*=T^*_+$. This implies that $T_{++}^*=T^*_+$. Moreover, $T(\alpha_1)\ge \max_{a_i\le 0}T^i(\alpha_1)$ by definition of $T$ hence $T_{++}^*=T^*_+\vee \max_{a_i\le 0}T^i(\alpha_1)$. The last statements are trivial so that  the conclusion follows.
\end{proof}

\begin{lemm} \label{Lemma-dual-affi-v1}
Let $K=[m,M]$ be a compact subset of $\R$,  $t\in [m,M]$ and  $a,b,c,d\in\R$ such that $a\le c$.  Let $f$ be a continuous function defined as
\bean
f(x)=(ax+b)1_{[m,t]}+(cx+d)1_{[t,M]}+ \infty 1_{R\setminus K} 
\eean
We have 
\bean
f^*(x)&:=& \left((x-a)m-b\right)1_{]-\infty,a]}+\left((x-a)t-b\right)1_{[a,c]}+\left((x-c)M-d\right)1_{]c,\infty]}. 
\eean
\end{lemm}

\begin{lemm}\label{Lemma-dual-affi-v2} Suppose that $\alpha_1<\alpha_2<\alpha_3$ and $T_1<T_2$. Consider a continuous function of the form
$$f(x)=(\alpha_1x+\beta_1)1_{x\le T_1}+(\alpha_2x+\beta_2)1_{T_1\le x\le T_2}+(\alpha_3x+\beta_3)1_{ x\ge T_2}.$$
Then,
\bean f^*(x)&=&\infty,\quad {\rm if\,} x<\alpha_1\, {\rm or}\, x> \alpha_3,\\
&=&(x-\alpha_1)T_1-\beta_1,\quad {\rm if\,} \alpha_1\le x\le \alpha_2,\\
&=&(x-\alpha_2)T_2-\beta_2,\quad {\rm if\,} \alpha_2\le x\le \alpha_3.
\eean
\end{lemm}

%%%%%%%%%%%%%%%%%%%%%% Subsection: C %%%%%%%%%%%%%%%%%%%%%%%
%%%%%%%%%%%%%%%%%%%%%%%%%%%%%%%%%%%%%%%%%%%%%%%%%%%%%%%%%%%%

\subsection{Proof of Section \ref{MR}} \quad
\begin{Proof} \label{Pr-Prop-Conv}Proof of Theorem \ref{theo:payoff}.\end{Proof}

Recall that, by Corollary \ref{coro-convex-infprice}, the infimum super hedging price at time $t-1$ is  \bean p_{t-1}(g_t)=\inf_{\alpha \geq 0, \, a \geq a_{t-1}(\alpha)}\varphi_{t-1}(\alpha,a,S_{t-1})\eean where, $\varphi_{t-1}(\alpha,a,x)=\hat\Phi_{\phi_{t-2}}(a)x+\bar y^1_{t-1}-am^1_{t-1}+\alpha.$  \smallskip
 
 Consider the first case where $S_{t-1}\in [m^{1+}_{t-1},m^{1-}_{t-1}]$. This is equivalent to \bea \label{C1}
 \alpha_{t-1}^1\leq 1+\kappa_{t-1},\quad
 \alpha_{t-1}^1\geq 1-\kappa_{t-1}.
 \eea
 In this case, i.e. under condition \eqref{C1}, Theorem \ref{slopes order} claims that the infimum super hedging price at time $t-1$ is  given by \bean p_{t-1}(g_t)=\inf_{\alpha \geq 0}\varphi_{t-1}(\alpha,a_{t-1}(\alpha)\vee\phi_{t-2},S_{t-1})\eean 
 
 %But, $\hat\Phi_{\phi_{t-2}}(a_{t-1}(\alpha)\vee\phi_{t-2})=(1+\kappa_{t-1})(a_{t-1}(\alpha)\vee\phi_{t-2})-\kappa_{t-1}\phi_{t-2}.$
 Le us introduce:
 \bean\delta_{\alpha}&:=&\hat\Phi_{\phi_{t-2}}(a_{t-1}(\alpha)\vee\phi_{t-2})S_{t-1}+\bar y^1_{t-1}-a_{t-1}(\alpha)\vee\phi_{t-2}m^1_{t-1}+\alpha,\\
 &=&\rho_{t-1}(a_{t-1}(\alpha)\vee\phi_{t-2})S_{t-1}+\hat g_t^1(\alpha_{t-1}S_{t-1})-\kappa_{t-1}\phi_{t-2}S_{t-1},
 \eean where
 \bea\label{def:rhocase1}
  \rho_{t-1}&=&(1+\kappa_{t-1})-\alpha^1_{t-1}.
  \eea Note that $\rho_{t-1}$ is positive under Condition \eqref{C1} and we have:
 \bean
\delta_{\alpha}
&=& \rho_{t-1}(a_{t-1}(\alpha)\vee\phi_{t-2})S_{t-1}+\hat g_t^1(\alpha_{t-1}S_{t-1})-\kappa_{t-1}\phi_{t-2}S_{t-1}\\
&=& \max_{i=1,\dots,2N_t}\psi_{t-1}^i(\alpha),
 \eean
  where
   \bean
\psi_{t-1}^1(\alpha)&=& \rho_{t-1}\phi_{t-2}S_{t-1}+\bar y_{t-1}^1-\kappa_{t-1}\phi_{t-2}S_{t-1}+w^1_{t-1}\alpha\\
   \psi_{t-1}^i(\alpha)&=& (1-w^i_{t-1})\bar y_{t-1}^i+w^i_{t-1}\bar y_{t-1}^1-\kappa_{t-1}\phi_{t-2}S_{t-1}+w^i_{t-1}\alpha,\quad i=2,\dots,N_t\\
       \psi_{t-1}^i(\alpha)&=& (1-w^i_{t-1})\bar Y_{t-1}^{i-N_t}+w^i_{t-1}\bar y_{t-1}^1-\kappa_{t-1}\phi_{t-2}S_{t-1}+w^i_{t-1}\alpha,\quad i=N_t+1,\dots,2N_t
\eean
with 
\bea\label{def:epsicase1}
   w^1_{t-1}&=&1,\\
   w^i_{t-1}&=&1-\frac{\rho_{t-1}}{\alpha_{t-1}^i-\alpha_{t-1}^1},\quad i=2,\dots,N_t,\\
    w^i_{t-1}&=&1-\frac{\rho_{t-1}}{\beta_{t-1}^{i-N_t}-\alpha_{t-1}^1},\quad i=N_t+1,\dots,2N_t. 
   \eea

   Note that under \eqref{C1}, $1-w^i_{t-1}$ is positive and $\delta_{\alpha}=\max_{i=1,\cdots,2N-1}\psi_{t-1}^i(\alpha)$. By virtue of Lemma \ref{KeyTool}, since $w^1_{t-1}=1>0$, we have:
   \bean
    p_{t-1}(g_t)=\max_{w^i_{t-1}\leq0,w^j_{t-1}>0}\psi_{t-1}^i(I_{ij})\vee\max_{w^j_{t-1}>0}\psi_{t-1}^j(0)
   \eean 
Clearly, $\max_{w^j_{t-1} > 0} \psi_{t-1}^j(0)$ is convex in $S_{t-1}$ since, for any $j \in \{1, \dots, 2N_t-1\}$, $\psi_{t-1}^j(0)$ is the sum of positive convex functions in $S_{t-1}$ provided that $w^j_{t-1}\ge 0$ and since $1 - w^j_{t-1}\ge 0$  under \eqref{C1}. At last, let us denote 
$$\tilde y_{t-1}^i=\tilde y_{t-1}^i(S_{t-1}):=\bar  y_{t-1}^i 1_{i\in\{1,\dots,N_t\}}+\bar  Y_{t-1}^{i-N_t}1_{i\in\{N_t+1,\dots,2N_t\}}.$$

Note that each  mapping $S_{t-1} \mapsto \tilde y_{t-1}^i(S_{t-1})$, $i=1,\cdots,2N_t$, is convex by assumption. Let us solve the equation $\psi_{t-1}^i(I_{ij}) = \psi_{t-1}^j(I_{ij})$. We get that  
\bean
I_{ij} &=& -\bar y_{t-1}^1 + \frac{(1 - w^i_{t-1})\tilde y_{t-1}^i - (1 - w^j_{t-1})\tilde y_{t-1}^j}{w^j_{t-1} - w^i_{t-1}},\,j\ne 1,\\
I_{i1} &=& -\bar y_{t-1}^1 +\tilde y_{t-1}^i-\frac{\rho_{t-1}}{w^1_{t-1}-w^i_{t-1}}\phi_{t-2}S_{t-1}.
\eean

Substituting this expression into $\psi_{t-1}^i(I_{ij})$, we obtain that: 
\bean
\psi_{t-1}^i(I_{ij}) =:g_{t-1}^{i,j}(\phi_{t-2},S_{t-1})&:=& \frac{(1 - w^i_{t-1})w^j_{t-1}}{w^j_{t-1}- w^i_{t-1}}\tilde y_{t-1}^i - \frac{(1 - w^j_{t-1})w^i_{t-1}}{w^j_{t-1} - w^i_{t-1}}\tilde y_{t-1}^j \\
&&- \kappa_{t-1}\phi_{t-2}S_{t-1}+\frac{-w^i_{t-1}}{1-w^i_{t-1}}\rho_{t}\phi_{t-2}S_{t-1}1_{\{j=1\}}.
\eean
Since $c^{i,j}_{t-1}:=\frac{(1 - w^i_{t-1})w^j_{t-1}}{w^j_{t-1} - w^i_{t-1}}=:\lambda^{i,j}_{t-1}\ge 0$ and $d^{i,j}_{t-1}=- \frac{(1 - w^j_{t-1})w^i_{t-1}}{w^j_{t-1} - w^i_{t-1}}\ge 0$, with $c^{i,j}_{t-1}+d^{i,j}_{t-1}=1$, we deduce that $g_{t-1}^{i,j}$ is a convex function of $S_{t-1}$. It is of the form $g_{t-1}^{i,j}(\phi_{t-2},x)=\hat g_{t-1}^{i,j}(x)-\hat \mu_{t-1}^{i,j}\phi_{t-2}x$ where $\hat \mu_{t-1}^{i,j}=\kappa_{t-1}$ if $j\ne 1$. Note that $\hat g_{t-1}^{i,j}(x)=c^{i,j}_{t-1}\tilde y_{t-1}^i(x)+d^{i,j}_{t-1}\tilde y_{t-1}^j(x)$, $x=S_{t-1}$. Otherwise, if $j=1$, $\hat \mu_{t-1}^{i,1}=\hat \mu_{t-1}^{i}=\kappa_{t-1}+\frac{w^i_{t-1}}{1-w^i_{t-1}}\rho_{t}$ so that $1+\hat \mu_{t-1}^{i}\in \{\alpha_{t-1}^i,\beta_{t-1}^i\}$. We deduce that $1+\hat \mu_{t-1}^{i,j}>0$ for any $i,j$. This means that each function $g_{t-1}^{i,j}$ is of type (\ref{GeneralPayoffForm1}) at time $t-1$ if $w_i\le 0$ and $w_j>0$.  Similarly, the functions $g_{t-1}^{1,j}(\phi_{t-2},x)=\psi_{t-1}^j(0)$ for $j$ such that $w^j>0$ are also of same type (\ref{GeneralPayoffForm1}) and we have either $\hat\mu_{t-1}^{j,1}=\kappa_{t-1}$ if $j\ne 1$ or $\hat\mu_{t-1}^{1,1}=\kappa_{t-1}-\rho_t=(1+\hat \mu_t^1)\alpha_{t-1}-1$ so that $1+\hat\mu_{t-1}^{j,1}>0$.  \smallskip
Since $c^{i,j}_{t-1}:=\frac{(1 - w^i_{t-1})w^j_{t-1}}{w^j_{t-1} - w^i_{t-1}}\ge 0$ and $d^{i,j}_{t-1}=- \frac{(1 - w^j_{t-1})w^i_{t-1}}{w^j_{t-1} - w^i_{t-1}}\ge 0$, we deduce that $g_{t-1}^{i,j}$ is a convex function of $S_{t-1}$. It is of the form $g_{t-1}^{i,j}(\phi_{t-2},x)=\hat g_{t-1}^{i,j}(x)-\hat \mu_{t-1}^{i,j}\phi_{t-2}x$ where $\hat \mu_{t-1}^{i,j}=\kappa_{t-1}$ if $j\ne 1$. Note that $\hat g_{t-1}^{i,j}(x)=c^{i,j}_{t-1}\tilde y_{t-1}^i(x)+d^{i,j}_{t-1}\tilde y_{t-1}^j(x)$, $x=S_{t-1}$. Otherwise, if $j=1$, $\hat \mu_{t-1}^{i1}=\hat \mu_{t-1}^{i}=\kappa_{t-1}+\frac{w^i_{t-1}}{1-w^i_{t-1}}\rho_{t}$ so that $1+\hat \mu_{t-1}^{i}\in \{\alpha_{t-1}^i,\beta_{t-1}^i\}$. We deduce that $1+\hat \mu_{t-1}^{i,j}>0$ for any $i,j$. This means that each function $g_{t-1}^{i,j}$ is of type (\ref{GeneralPayoffForm1}) at time $t-1$ if $w_i\le 0$ and $w_j>0$.  Similarly, the functions $g_{t-1}^{1,j}(\phi_{t-2},x)=\psi_{t-1}^j(0)$ for $j$ such that $w_{t-1}^j>0$ are also of same type (\ref{GeneralPayoffForm1}) and we  either have $\hat\mu_{t-1}^{j,1}=\kappa_{t-1}$ if $j\ne 1$ or $\hat\mu_{t-1}^{1,1}=\kappa_{t-1}-\rho_t=\alpha_{t-1}^1-1$ so that $1+\hat\mu_{t-1}^{j,1}>0$.  \smallskip

Notice that the set of all $(i,j)$ such that $w_{t-1}^i\le 0$ and $w_{t-1}^j>0$ does not depend on $S_{t-1}$ nor $\phi_{t-2}$. The same holds for the set of all $j$ such that $w_{t-1}^j>0$. Therefore, we conclude that $p_{t-1}(g_t)$ is a convex function of $S_{t-1}$,  is of type (\ref{GeneralPayoffForm1}) at time $t-1$ and satisfies \eqref{formGi-11}–\eqref{formGi1}.\bigskip

Consider the second case where $S_{t-1}\in [m_{t-1}^{1-},M_{t-1}^{N-}]$ which is equivalent to the condition
\bea\label{C2}
\beta_{t-1}^{N_t}\geq 1-\kappa_{t-1},\quad \alpha_{t-1}^1\leq 1-\kappa_{t-1}.
\eea
By Theorem \ref{optimization} and Corollary \ref{coro-convex-infprice}, the infimum super hedging price under \eqref{C2} at time $t-1$ is  \bean p_{t-1}(g_t)&=&\inf_{\alpha \geq 0}\varphi_{t-1}(\alpha,a_{t-1}(\alpha),S_{t-1})
=\inf_{\alpha \geq 0}\delta_{\alpha}
\eean
where 
\bean
\delta_{\alpha}&=&\hat\Phi_{\phi_{t-2}}(a_{t-1}(\alpha))S_{t-1}+\bar y^1_{t-1}-a_{t-1}(\alpha)m^1_{t-1}+\alpha\\
&=& \max_{i=1,2}\left( \rho_t^ia_{t-1}(\alpha)S_{t-1}+\bar y^1_{t-1}+(-1)^i\kappa_{t-1}\phi_{t-2}S_{t-1}+\alpha\right)\\
&=& \max_{i=1,2}\max_{j=2,\dots,2N_t}\left((1-w_{t-1}^{(i,j)})\tilde y_{t-1}^j+w_{t-1}^{(i,j)}\bar y_{t-1}^1+w_{t-1}^{(i,j)}\alpha+(-1)^i\kappa_{t-1}\phi_{t-2}S_{t-1}\right)\\
&=:& \max_{i=1,2}\max_{j=2,\dots,2N_t}\psi_{t-1}^{i,j}(\alpha)
\eean
where, for $r=1,2$,
\bea\label{epsiloncase2}
 w_{t-1}^{(r,j)}&=&1-\frac{\rho_t^r}{\alpha_{t-1}^j-\alpha_{t-1}^1},\quad j=2,\dots,N_t\\
    w_{t-1}^{(r,j)}&=&1-\frac{\rho_t^r}{\beta_{t-1}^{j-N_t}-\alpha_{t-1}^1},\quad j=N_t+1,\dots,2N_t\notag\\
 \rho_t^r&=&(1-(-1)^r \kappa_{t-1})-\alpha_{t-1}^1.\notag\eea
Note that under condition \eqref{C2}, both $\rho_t^1$ and $\rho_t^2$ are positive and so $1 - w_{t-1}^{(1,j)}$ and $1 - w_{t-1}^{(2,j)}$ are.

Under Condition \eqref{C2}, $w_{t-1}^{(2,2N_t)}\geq 0$. Let us first suppose that $w_{t-1}^{(2,2N_t)}>0$
By virtue of Lemma \ref{KeyTool}, we have:
   \bean
    p_{t-1}(g_t)=\max_{w_{t-1}^{(r,i)}\leq0,w_{t-1}^{(m,j)}>0}\psi_{t-1}^{r,i}(I_{ij}^{r,m})\vee\max_{w^{m,j}>0}\psi_{t-1}^{m,j}(0)
   \eean 
 As in the first case, $\max_{w_{t-1}^{(m,j)}>0}\psi_{t-1}^{m,j}(0)$ is convex in $S_{t-1}.$ Moreover, each function $\psi_{t-1}^{m,j}(0)$ is of the form
 $\psi_{t-1}^{m,j}(0)=\hat g_{t-1}^{1,m,j}(S_{t-1})-\hat\mu_{t-1}^{(m,j)}\phi_{t-2}S_{t-1}$ where the coefficent $\hat\mu_{t-1}^{(m,j)}=(-1)^{m}\kappa_{t-1}$ satisfies $1+\hat\mu_{t-1}^{(m,j)}>0$. This means that 
 $\psi_{t-1}^{m,j}(0)$ is a function of $\phi_{t-2}$ and $S_{t-1}$ that satisfies the property (\ref{GeneralPayoffForm1}) at time $t-1$ under the condition (\ref{formGi1}).\bigskip

 Let us solve the equation $\psi_{t-1}^{r,i}(I_{ij}^{k,m}) = \psi_{t-1}^{m,j}(I_{ij}^{r,m})$. We obtain that
 \bean
 I_{ij}^{r,m}=-\bar y_{t-1}^1+\frac{(1-w_{t-1}^{(m,j)})\tilde y^j-(1-w_{t-1}^{(r,i)})\tilde y_{t-1}^i}{w_{t-1}^{(r,i)}-w_{t-1}^{(m,j)}}+\frac{(-1)^m-(-1)^r}{w_{t-1}^{(r,i)}-w_{t-1}^{(m,j)}}.
 \eean 
 We deduce that: 
 \bea\label{psicase2}
\psi_{t-1}^{r,i}(I_{ij}^{r,m})&=&\frac{w_{t-1}^{(m,j)}(1-w_{t-1}^{(r,i)})}{w_{t-1}^{(m,)j}-w_{t-1}^{(r,i)}}\tilde y_{t-1}^i+\frac{-w_{t-1}^{(r,i)}(1-w_{t-1}^{(m,j)})}{w_{t-1}^{(m,j)}-w_{t-1}^{(r,i)}}\tilde y_{t-1}^j\\
&&+\frac{w_{t-1}^{(m,j)}(-1)^k-w_{t-1}^{(r,i)}(-1)^m}{w_{t-1}^{(m,j)}-w_{t-1}^{(r,i)}}\kappa_{t-1}\phi_{t-2}S_{t-1}.
 \eea
 
 When $w_{t-1}^{(r,i)}\leq 0$ and $w_{t-1}^{(m,j)}>0$, $\psi_{t-1}^{k,i}(I_{ij}^{k,m})$ is a convex function of $S_{t-1}$ by assumption of $\hat g_t^i$ and $\hat g_t^j$ defining $\tilde y_{t-1}^i$ and $\tilde y_{t-1}^j$ respectively. Precisely, we may write $\psi_{t-1}^{r,i}(I_{ij}^{r,m})=\hat g_{t-1}^{r,m,i,j}(S_{t-1})-\hat\mu_{t-1}^{(r,m,i,j)}\phi_{t-2}S_{t-1}$ where $\hat g_{t-1}^{r,m,i,j}$ is a convex function and $\hat\mu_{t-1}^{(r,m,i,j)}=\frac{w_{t-1}^{(r,i)}(-1)^m-w_{t-1}^{(m,j)}(-1)^r}{w_{t-1}^{(m,j)}-w^{(r,i)}}\kappa_{t-1}$ satisfies $1+\hat\mu_{t-1}^{(r,m,i,j)}>0$ if $w^{(m,j)}>0$ and $w_{t-1}^{(r,i)}\le 0$. We deduce that $\psi_{t-1}^{r,i}(I_{ij}^{r,m})=\psi_{t-1}^{r,i}(I_{ij}^{r,m})(\phi_{t-2},S_{t-1})$ is of type (\ref{GeneralPayoffForm1}) at time $t-1$ and satisfying \eqref{formGi-11}–\eqref{formGi1}.\smallskip
 
 In the case where $w_{t-1}^{(2,2N_t)}=0$, we have $p_{t-1}(g_t)=\psi_{t-1}^{m,j}(I_{j,2N_t}^{m,2})$ by  Lemma \ref{KeyTool} where $m\in\{1,2\}$ and $j\in\{2,\cdots,2N_t-1\}$ are such that $\epsilon^{m,j}$ is the largest negative slope. Therefore, we conclude similarly by (\ref{psicase2}) that $p_{t-1}(g_t)$ is a convex function of $S_{t-1}$. In particular, the coefficient $\hat \mu_{t-1}^{(m,2,j,2N_t)}=(-1)^m\kappa_{t-1}$ so that $1+\hat \mu_{t-1}^{(m,2,j,2N_t)}>0$.
\fdem

\begin{Proof} \label{Pr-OptStrat11} Proof of Proposition \ref{OptStrat11}.
\end{Proof}

   From Corollary \ref{coro-convex-infprice} and Theorem \ref{optimization} we have that  
\bean
p_{t-1}(g_t)(x) &= \inf_{\alpha\geq 0}\varphi_{t-1}(\alpha,a_{t-1}(\alpha)\vee\phi_{t-2},x), \quad x \in [m_{t-1}^{1+},m_{t-1}^{1-}].
\eean
Recall that the mappings $\varphi_{t-1}$ and $\hat\Phi_{\phi_{t-2}}$ are defined as    
\bean
\varphi_{t-1}(\alpha,a,x) &=& \hat\Phi_{\phi_{t-2}}(a)x + \bar y_{t-1}^1 - a m_{t-1}^1 + \alpha,\\
\hat\Phi_{\phi_{t-2}}(a)&=&a+\kappa_{t-1}|a-\phi_{t-2}|.
\eean
Here $x=S_{t-1}$ and recall that $m_{t-1}^{1\pm}=\alpha^1_{t-1}S_{t-1}(1\pm \kappa_{t-1})^{-1}$. Therefore, the condition  $x \in [m_{t-1}^{1+},m_{t-1}^{1-}]$ means that $ \alpha_{t-1}^1 \geq 1-\kappa_{t-1}$ and, also, the AIP condition is satisfied. By Theorem \ref{slopes order}, the above infimum is attained, i.e.  there exits $\alpha^*\ge 0$ such that
\bea \label{Minimal-price-General-case-1}  
    p_{t-1}(g_t) &=& \hat{g}_{t}^1(\alpha_{t-1}S_{t-1}) + a_{t-1}(\alpha^*) \vee \phi_{t-2}(1+\kappa_{t-1}-\alpha_{t-1}^1)S_{t-1} - \kappa_{t-1} \phi_{t-2} S_{t-1} + \tilde\alpha^*, \quad \quad,\quad\\
    p_{t-1}(g_t) &=&A_{t-1}^{(1)}(\alpha^*)+\hat{g}_{t}^1(\alpha_{t-1}S_{t-1})- \kappa_{t-1} \phi_{t-2} S_{t-1},
    \eea
    As $A_{t-1}^{(1)}$ is a piecewise affine convex function in $\alpha$,
      we may choose  the argmax $\alpha^*\in \arg\min_{\alpha\geq0} A_{t-1}^{(1)}(\alpha)$ such that $\alpha^* \in  \left(\arg\min_{e\in I}|A_{t-1}^{(1)}(e)-A_{t-1}^{(1)}(\alpha^*)|\right)^+$, i.e.
    \bean
    &&\alpha^* \in  \left(\arg\min_{e\in I}|A_{t-1}^{(1)}(e)-p_{t-1}(g_t)+\hat{g}_{t}^1(\alpha_{t-1}S_{t-1})- \kappa_{t-1} \phi_{t-2} S_{t-1}|\right)^+.
\eean

Let us establish that the super-replication property holds with the optimal strategy $\phi_{t-1}^{opt} = a_{t-1}(\alpha^*)\vee \phi_{t-2}$. Notice that $a_{t-1}(\alpha^*)\ge \phi_{t-2}$ if $\phi_{t-2}<a_{t-1}(0)$, see Proof \ref{optimization}. With $V_{t-1}=p_{t-1}(g_t)$ and $\tilde\alpha^*=\alpha^*1_{\{\phi_{t-2}<a_{t-1}(0)\}}$, the self-financing portfolio $(V_u)_{u=t-1,t}$ satisfies  
\bean
V_t &=& V_{t-1} + \phi_{t-1}^{opt} \Delta S_t - \kappa_{t-1} (\phi_{t-1}^{opt} - \phi_{t-2}) S_{t-1}\\
&=& \hat{g}_{t}^1(\alpha_{t-1}S_{t-1}) + \phi_{t-1}^{opt} (S_t- \alpha_{t-1}^1 S_{t-1}) + \tilde\alpha^*,\\
&=& \hat{g}_{t}^1(\alpha_{t-1}S_{t-1}) + \phi_{t-1}^{opt} (S_t - m_{t-1}^1) + \tilde\alpha^*,\\
&=&aS_t+b=:\gamma_{a,b}(S_t),
\eean 
where the coefficients
$a = \phi_{t-1}^{opt}$ and $b = \hat{g}_{t}^1(\alpha_{t-1}S_{t-1}) - \phi_{t-1}^{opt} m_{t-1}^1 + \tilde\alpha^*$ satisfy the inequalities  $ \gamma_{a,b}(x) \geq \bar g_{t}^i(x)$, for all  $x\in K_{t-1}^i$, $i=1,\cdots,N$, see   Proof of Theorem \ref{hbar}. Therefore, replacing $x\in K_{t-1}^i$ by $x(1+\hat\mu_t^i)$ where $x\in C_{t-1}$, we get that 
\bean
\gamma_{a,b}(x) &\ge& \hat g_t^i(x) - \hat\mu_t^i \phi_{t-1}^{opt} x,  \quad i = 1, \dots, N,\\
&\geq& \max_{i=1, \dots, N} \big( \hat g_t^i(x) - \hat\mu_t^i \phi_{t-1}^{opt} x \big), \quad x\in  C_{t-1}.
\eean 
Since $S_t \in C_{t-1}={\rm supp}_{\cF_{t-1}}S_t$ a.s. by definition of the conditional support, we deduce the desired inequality 
$V_t\ge \max_{i=1, \dots, N} \big( \hat g_t^i(S_t) - \hat\mu_t^i \phi_{t-1}^{opt} S_t \big)=g_t(\phi_{t-1}^{opt}, S_t)$.
\fdem

\begin{Proof}\label{Pr-OptStrat21}
    Proof of Proposition \ref{OptStrat21}.
\end{Proof}

     By Theorem \ref{slopes order}, the infimum super-hedging price is attained and given by
     \bean
     p_{t-1}(g_t)&=& \hat{g}_{t}^1(\alpha_{t-1}S_{t-1})+a_{t-1}(\alpha^*)(1-\alpha_{t-1}^1)S_{t-1}+\kappa_{t-1}|a_{t-1}(\alpha^*)-\phi_{t-2}|S_{t-1}+\alpha^*,\\
     &=&\hat{g}_{t}^1(\alpha_{t-1}S_{t-1})+A_{t-1}^{(2)}(\alpha^*),
     \eean
     where, as in  Proof \ref{Pr-OptStrat11}, we have $$\alpha^*\in \left(\arg\min_{e\in I_{t-1}}|A_{t-1}^{(2)}(e)-p_{t-1}(g_t)+\hat{g}_{t}^1(\alpha_{t-1}S_{t-1})|\right)^+\subseteq \arg\min_{\alpha\geq 0}A_{t-1}^{(2)}(\alpha).$$
     Let us establish that the super-replication property holds with the optimal strategy $\phi_{t-1}^{opt} = a_{t-1}(\alpha^*)$. With $V_{t-1}=p_{t-1}(g_t)$, the self-financing portfolio $(V_u)_{u=t-1,t}$ satisfies 
\bean
V_t &=& V_{t-1} + \phi_{t-1}^{opt} \Delta S_t - \kappa_{t-1} |\phi_{t-1}^{opt} - \phi_{t-2}| S_{t-1}\\
&=& \hat{g}_{t}^1(\alpha_{t-1}S_{t-1}) + \phi_{t-1}^{opt} (S_t- \alpha_{t-1}^1 S_{t-1}) + \alpha^*,\\
&=& \hat{g}_{t}^1(\alpha_{t-1}S_{t-1}) + \phi_{t-1}^{opt} (S_t - m_{t-1}^1) + \alpha^*,\\
&=&aS_t+b=:\gamma_{a,b}(S_t),
\eean 
where the coefficients
$a = \phi_{t-1}^{opt}$ and $b = \hat{g}_{t}^1(\alpha_{t-1}S_{t-1}) - \phi_{t-1}^{opt} m_{t-1}^1 + \tilde\alpha^*$ satisfy the inequalities  $ \gamma_{a,b}(x) \geq \bar g_{t-1}^i(x)$, for all  $x\in K_{t-1}^i$, $i=1,\cdots,N$, see   Proof \ref{hbar}. Therefore, replacing $x\in K_{t-1}^i$ by $x(1+\hat\mu_t^i)$ where $x\in C_{t-1}$, we get that 
\bean
\gamma_{a,b}(x) &\ge& \hat g_t^i(x) - \hat\mu_t^i \phi_{t-1}^{opt} x,  \quad i = 1, \dots, N,\\
&\geq& \max_{i=1, \dots, N} \big( \hat g_t^i(x) - \hat\mu_t^i \phi_{t-1}^{opt} x \big), \quad x\in  C_{t-1}.
\eean 
Since $S_t \in C_{t-1}={\rm supp}_{\cF_{t-1}}S_t$ a.s. by definition of the conditional support, we deduce the desired inequality 
$V_t\ge \max_{i=1, \dots, N} \big( \hat g_t^i(S_t) - \hat\mu_t^i \phi_{t-1}^{opt} S_t \big)=g_t(\phi_{t-1}^{opt}, S_t)$.
\fdem

\noindent   {\bf Data Availability Statement: The data that support the findings of this study are available in \url{https://finance.yahoo.com/quote/SPY/history.}

% ===== Bibliography =====
\bibliographystyle{plain} 

\begin{thebibliography}{100}

%\bibliography{biblio} 

\bibitem{Albanese06}
 Albanese C. Small transaction cost asymptotics and dynamic hedging. European Journal of Operational Research, 185, 3, 1404, 1414, 2008.



\bibitem{Biagini23}
 Biagini F. Neural network approximation for superhedging prices. Mathematical Finance, 33, 2, 123–145, 2023.

\bibitem{BSV03}
 Bouchard B.,  Schweizer M, and  Vu L. Discrete-time super-replication under proportional transaction costs. Mathematical Finance, 13, 3, 363–386, 2003.

\bibitem{GM03}
 Gobet E. and  Miri M. Discrete-time hedging with transaction costs: convergence and numerical methods. Finance and Stochastics, 7, 2, 153–182, 2003.

\bibitem{CS} Campi L. and Schachermayer W. A super-replication theorem in Kabanov's model for transaction costs". Finance and Stochastics, 10, 4, 2006.

\bibitem{CL} Carassus L. and  Lepinette E. Pricing without
no-arbitrage condition in discrete time. J. Math. Anal. Appl., 505, 1, 21, 2022.

\bibitem{DKL}De Vallière D., Lepinette E. and Kabanov Y. Hedging of American options under transaction costs. Finance and Stochastics 13, 1, 105-119, 2009.

\bibitem{FK97}
 Föllmer F. and  Schweizer M. Hedging of contingent claims under transaction costs. Mathematical Finance, 7, 2, 117–145, 1997.


\bibitem{ELM}El Mansour M. and  Lepinette E. Conditional interior and conditional closure of a random sets. Journal of Optimization Theory and Applications, 187, 356-369, 2020.

\bibitem{ElMansour22}
 El Mansour M. and  Lepinette E. Robust discrete-time super-hedging strategies under AIP condition. MSIA, MathematicS In Action,11, 193-212, 2022.

\bibitem{GLR} Guasoni P., Lepinette E. and   Rasonyi M. The fundamental theorem of asset pricing under transaction costs.  Finance and Stochastics. 16, 4, 741-777, 2012. 

\bibitem{KS} Kabanov Y. and Safarian M. On Leland’s strategy of option pricing with transactions costs. Finance and Stochastics, 1, 239-250, 1997.

\bibitem{KL}Kabanov Y. and  Denis(Lepinette) E. Mean square error for the Leland-Lott hedging strategy: convex pay-off. Finance and Stochastics, 14, 4, 626-667, 2010.

\bibitem{Koehl99}
 Koehl P.F. On super-replication in discrete time under transaction costs. Mathematical Finance, 9, 3, 207–227, 1999.


\bibitem{Koehl01}
 Koehl P.F. On super-replication in discrete time under transaction costs. SIAM Journal on Control and Optimization, 39, 3, 839–860, 2001.

\bibitem{Leland85}
 Leland H. Option pricing and replication with transaction costs. Journal of Finance, 40, 5, 1283–1301, 1985.

\bibitem{Lep1}Lepinette E.
Modified Leland's strategy for constant transaction costs rate. Mathematical Finance, 22, 4, 741-752, 2012. 

\bibitem{Lep0} Denis(Lepinette) E. Approximate hedging of contingent claims under transaction costs.  Applied Mathematical Finance, 17, 491-518, 2010.

\bibitem{Lep2} Lepinette E., Vu D.T. Dynamic programming principle and computable prices in financial market models with transaction costs. Journal of Mathematical Analysis and Applications, 524, 127068, 2023.



\bibitem{Perg} Serguei Pergamenshchikov. Limit theorem for Leland's strategy. Annals of Applied Probability, 13, 3, 1099-1118, 2003. 

\bibitem{rw} Rockafellar R.T.,  Wets R.J.B. and  Wets M. Variational analysis. In Grundlehren der mathematischen Wissenschaften, 1998.

\bibitem{Rouge}
Rouge R. and  El Karoui N. Hedging option portfolios in the presence of transaction costs. Mathematical Finance, 13, 1, 55–79, 2003.

\bibitem{PT1}Pergamenshchikov S. and Thai N. Approximate hedging with constant proportional transaction costs in financial markets with jumps. Theory of Probability and Its Applications, 65, 2, 224-248, 2020.

\bibitem{PT2}Pergamenshchikov S. and Thai S.
Approximate hedging problem with transaction costs in stochastic volatility markets. Mathematical Finance, 27, 3, 832-865, 2017.

\bibitem{VZ} Vinter R.B. and Zheng H. Some finance problems solved with nonsmooth optimization techniques, JOTA, 119, 1, 2003.

\bibitem{WG} Guasoni P. and Weber M.H. Rebalancing multiple assets with mutual price impact. Journal of Optimization Theory and Applications, 179, 2,  618-653, 2018.
\end{thebibliography}

\end{document}